\documentclass[12pt,fleqn]{article}
\usepackage[T1]{fontenc}
\usepackage[utf8]{inputenc}
\usepackage[english]{babel}
\usepackage{amssymb,amsmath}
\usepackage{amsthm}
\usepackage{cancel}
\usepackage{palatino}
\usepackage{color}
\usepackage{longtable}
\usepackage{mathdots}
\usepackage{enumerate}

\allowdisplaybreaks

\numberwithin{equation}{section}

\newtheorem{thm}{Theorem}[]
\newtheorem{coro}[thm]{Corollary}
\newtheorem{prop}[thm]{Proposition}

\theoremstyle{definition}
\newtheorem*{rmk}{Remark}

\theoremstyle{definition}

\def\d{\partial}
\def\f{\frac}

\def\p{\partial}
\newcommand{\eqa}{\begin{eqnarray}}
\newcommand{\eeqa}{\end{eqnarray}}
\newcommand{\beq}{\begin{equation}}
\newcommand{\eeq}{\end{equation}}
\def\proof{\begin{center} {\bf Proof:} \end{center}\vspace{0.5pt}}
\def\QEDclosed{\mbox{\rule[0pt]{1.3ex}{1.3ex}}} 
\def\QED{\QEDclosed} 
\def\endproof{\hspace*{\fill}~\QED\par\endtrivlist\unskip}

\begin{document}
\title{Deformations of non semisimple Poisson pencils of hydrodynamic type}
\author{Alberto Della Vedova ${}^{*}$, Paolo Lorenzoni ${}^{*}$ and Andrea Savoldi ${}^{**}$}
    \date{}
    \maketitle
    \vspace{-7mm}
\begin{center}
${}^{*}$ Dipartimento di Matematica e Applicazioni, University of Milano-Bicocca\\
via Roberto Cozzi 53 I-20125 Milano, Italy\\
${}^{**}$ Department of Mathematical Sciences, Loughborough University \\
Leicestershire LE11 3TU, Loughborough, United Kingdom \\
e-mails: \\[1ex] 
\texttt{alberto.dellavedova@unimib.it}\\
\texttt{paolo.lorenzoni@unimib.it}\\
\texttt{A.Savoldi@lboro.ac.uk}\\
\end{center}

\bigskip

\begin{abstract}
We study deformations of two-component non semisimple Poisson pencils of hydrodynamic type associated with  Balinski\v{\i}-Novikov algebras. We show that in most cases the second order deformations are parametrized by two functions of a single variable. It turns out that one function is invariant with respect to the subgroup of Miura transformations preserving the dispersionless limit and another function is related to a one-parameter family of truncated structures. In two expectional cases the second order deformations are parametrized by four functions. Among them two are invariants and two are related to a two-parameter family of truncated structures. We also study the lift of deformations of n-component semisimple structures. This example suggests that deformations of non semisimple pencils corresponding to the lifted invariant parameters are unobstructed.

\bigskip


\bigskip

\end{abstract}

\tableofcontents

\section{Introduction}
Poisson pencils of hydrodynamic type and their deformations play an important role in the modern theory of integrable PDEs.
 Originally the study of such structures was motivated  by questions arising in the theory of Frobenius manifolds, Gromov-Witten invariants and topological field theory \cite{D,DZ}. In this setting, the deformations satisfy some additional constraints ($\tau$-structure, Virasoro constraints) and the undeformed pencil is related to a Frobenius manifold \cite{D}.

A perturbative approach to the study of these deformations was developed by Dubrovin and Zhang in \cite{DZ}. In their approach, the full pencil
\begin{equation}\label{pencil}
\begin{array}{l}
\displaystyle
\Pi^{ij}_{\lambda}=\omega^{ij}_{2}+\sum_{k\ge 1}\epsilon^k\sum_{l=0}^{k+1}A^{ij}_{2;k,l}(u,u_x,\dots,u_{(l)})\d_x^{(k-l+1)}\\
\displaystyle
\qquad
-\lambda\left(\omega^{ij}_{1}+\sum_{k\ge 1}\epsilon^k\sum_{l=0}^{k+1}A^{ij}_{1;k,l}(u,u_x,\dots,u_{(l)})\d_x^{(k-l+1)}\right),
\end{array}
\end{equation}
($A^{ij}_{\theta;k,l}$ are homogeneous differential polynomials of degree $l$) is obtained via a bi-Hamiltonian deformation procedure from the dispersionless limit $\epsilon\to 0$:
\beq\label{eq::hdpencil}
\omega^{ij}_{2}-\lambda\omega^{ij}_{1}=g_{2}^{ij} \d_x+b^{ij}_{2;k} u^k_{x}-\lambda\left(g_{1}^{ij} \d_x+b^{ij}_{1;k} u^k_{x}\right).
\eeq
The pencil of metrics \cite{D,F} $g_{\lambda}=g_{2}-\lambda g_{1}$ defining this limit is assumed to be semisimple, meaning that  
there exists a special set of coordinates, the roots $(r^1,...,r^n)$ of the equation
 $\det g_{\lambda}=0$, such that both metrics
 of the pencil $g_{\lambda}$ take diagonal form.

Two deformations $\Pi_{\lambda}$ and $\tilde\Pi_{\lambda}$ of the same pencil are considered equivalent if they are related by
a  Miura transformation of the form
 \beq\label{Miura}
\tilde{u}^i= u^i+\sum_{k\ge1}\epsilon^k F^i_k(u,u_x,\dots,u_{(k)}),
\eeq
where  $F^i_k(u,u_x,\dots,u_{(k)})$ are differential polynomials of degree $k$. This means that two pencils  belonging to the same class are related by
$$
\tilde\Pi_{\lambda}^{ij}=L^{*i}_k\Pi_{\lambda}^{kl}L^j_l,
$$
where
$$L^i_k=\sum_s(-\d_x)^s\f{\d\tilde{u}^i}{\d u^{(k,s)}},\qquad L^{*i}_k=\sum_s\f{\d\tilde{u}^i}{\d u^{(k,s)}}\d_x^s.$$
Dubrovin, Liu and Zhang proved that the equivalence classes are labelled by $n$ functions $c^i(r^i)$ called \emph{central invariants} \cite{LZ,DLZ}. These functions are obtained by expanding the roots $\lambda^i$ of the equation
$$
{\rm det}\left(g^{ij}_{2}-\lambda g^{ij}_{1}+\sum_{k\ge 1}\left(A^{ij}_{2;k,0}(u)-\lambda A^{ij}_{1;k,0}(u)\right)p^{k}\right)=0,
$$
 near $\lambda^i=r^i$:
\begin{equation}\label{expansion}
\lambda^i=r^i+\sum_{k=1}^{\infty}\lambda^i_{2k} p^{2k},
\end{equation}
and selecting the coefficient of $p^2$. The central invariants are then defined as \cite{DLZ2,LZ}:
$$c^i=\f{1}{3}\f{\lambda^i_2}{g_1^{ii}}=\frac{1}{(f^i)^2} \left(Q^{ii}_2-r^i Q^{ii}_1+\sum_{k\neq i} \frac{(P^{ki}_2-r^i P^{ki}_1)^2}{f^k (r^k-r^i)} \right), \quad i=1,\ldots, n,$$
where $f^i$ are the diagonal components of the contravariant metric $g_1$ in canonical coordinates and
$$
P^{ij}_{\theta}(u)=A^{ij}_{\theta;1,2}(u), \quad Q^{ij}_{\theta}(u)=A^{ij}_{\theta;2,3}(u),\quad i,j=1,\ldots, n, \quad \theta=1,2. 
$$
They can also be defined by (see \cite{FL})
\begin{equation*}
c_i=-\frac{1}{3 f^i} \mathrm{Res}_{\lambda=r^i} \mathrm{Tr} \left[g_{\lambda}^{-1}(Q^{ij}_{\lambda}+(g_{\lambda}^{-1})_{lk} P^{li}_\lambda P^{kj}_{\lambda})\right],
\end{equation*}
where $Q^{ij}_{\lambda}=Q^{ij}_{2}-\lambda Q^{ij}_{1}$ and $ P^{ij}_{\lambda}=P^{ij}_{2}-\lambda P^{ij}_{1}.$
\newline
\newline
In this framework the following facts should be mentioned:
\begin{itemize}
\item Each function $c^i$ depends only on the corresponding canonical coordinate $r^i$ and it is invariant with respect to Miura transformations \eqref{Miura} \cite{LZ}. 
\item Two deformations (of the same pencil) belong to the same class of equivalence if and only if they have the same central invariants \cite{DLZ}.
\item For any choice of the dispersionless limit and of the central invariants the equivalence classes are not empty. This fact, suggested by some computations (for the scalar case see \cite{L,AL}),
   has been proved only recently: by Liu and Zhang  in the scalar case \cite{LZcoho} and  by Carlet, Posthuma and Shadrin  in the general semisimple case \cite{CPS}. The proof is based on the vanishing of certain cohomology groups introduced in \cite{LZ}.
\item Fixed the dispersionless limit $\omega_{\lambda}$ and the central invariants $c^i(r^i)$ there exists a Miura transformation
 \eqref{Miura} reducing the pencil to the standard form \cite{LZ}
\begin{eqnarray*}
\Pi_{\lambda}&=&\omega_{2}-\lambda\omega_{1}+\epsilon^2 {\rm Lie}_{X_{(c_1,..,c_n)}} \omega_{1}+\epsilon^4 \Pi_4+\epsilon^6 \Pi_6+...\\
&=&\omega_{2}-\lambda\omega_{1}+\epsilon^2 {\rm Lie}_{Y_{(c_1,..,c_n)}} \omega_{2}+\epsilon^4 \Pi_4+\epsilon^6 \Pi_6+...
\end{eqnarray*}
where the \emph{polynomial} vector fields $X_{(c_1,...,c_n)}$ and $Y_{(c_1,...,c_n)}$  can be written as difference of two Hamiltonian vector fields 
$$X_{(c_1,...,c_n)}=\omega_{2}\,\delta H-\omega_{1}\,\delta K,\qquad Y_{(c_1,...,c_n)}=\omega_{2}\,\delta H'-\omega_{1}\,\delta K'$$ 
with \emph{non polynomial} hamiltonian densities:
\begin{eqnarray}\label{H}
H[r]&=&\sum_{i=1}^n\int c^i(r^i)r^i_x{\rm log}r^i_x\,dx,\quad K[r]=\sum_{i=1}^n\int r^i c^i(r^i)r^i_x{\rm log}r^i_x\,dx.\\
H'[r]&=&\sum_{i=1}^n\int \f{c^i(r^i)}{r^i}r^i_x{\rm log}r^i_x\,dx,\quad  K'[r]=\sum_{i=1}^n\int c^i(r^i)r^i_x{\rm log}r^i_x\,dx.
\end{eqnarray}
\item The coefficients $F_k(u,u_x,\dots,u_{(k)})$ of the Miura transformation \eqref{Miura} are assumed to depend  \emph{polynomially} on the derivatives of $u^i$.
 Removing this assumption the classification problem becomes "trivial": all deformations turn out to be equivalent to their dispersionless limit.
 This remarkable property of the deformations was discovered in \cite{DLZ} and it is called
 \emph{quasitriviality}. For instance, it is easy to check that the \emph{canonical quasi-Miura transformation} generated by the Hamiltonian $H$ defined in the formula \eqref{H} reduces the pencil $\Pi^{ij}_\lambda$ to the form $\omega^{ij}_{2}-\lambda\omega^{ij}_{1}+\mathcal{O}(\epsilon^4).$
\end{itemize}

In the present paper we start the study of the non semisimple case. 
Whereas the semisimple case is fairly understood, the non semisimple case is widely open.
Beside computational difficulties, the lack of canonical coordinates, or at least of a normal form theorem for non semisimple pencils, makes very hard a unified approach to the problem. For this reason  in this paper
 we try and get some information on the general case focusing on two special subcases where
  computations are feasible:
\newline
\newline
\emph{The deformations of Poisson pencils related to two-dimensional Balinski\v{\i}-Novikov algebras \cite{BN} and the associated invariant bilinear forms}.
 These are two component Poisson pencils that can be reduced to the form 
$$
\omega^{ij}_{2}-\lambda\omega^{ij}_{1}=g^{ij} \d_x+b^{ij}_{k} u^k_{x}-\lambda\eta^{ij} \d_x
$$
where $g^{ij}$ depends linearly on the variable $(u^1,...,u^n)$ and the coefficients $b^{ij}_{k}$ and $\eta^{ij}$ are constant.  Special deformations associated with second and third order cocycles of Balinski\v{\i}-Novikov algebras naturally arise in the study of multi-component generalizations of the Camassa-Holm equation \cite{SS}. We will consider deformations of two component non degenerate structures related to Balinski\v{\i}-Novikov algebras, that is the cases T3, N3, N4 (for $\eta^{11}=0$), N5 and N6 (for $\kappa\neq -1$) of the Bai-Meng's list \cite{BM} (which is recalled afterwards in Section \ref{sect_BN}, Table \ref{tab1}). The cases N1 and N4 with $\eta^{11}\ne 0$ are semisimple and then they are covered by Dubrovin-Liu-Zhang theory.
The non semisimple structures we focus on are summarized on the next table, where we also write down the corresponding affinor $L=g \eta^{-1}$.

\newpage

\LTcapwidth=\textwidth
\arraycolsep=4.5pt
\setlength{\tabcolsep}{3.5pt}
{\footnotesize
\begin{longtable}{ c  c  c  c  c  c }
\caption{Pair of metrics of bi-Hamiltonian structures.}\label{tab4}\\
 \hline
Type & \begin{tabular}{c} Linear metrics \\ $g$ \end{tabular}&  \begin{tabular}{c}Constant metrics\\ $\eta$ \end{tabular} & \begin{tabular}{c}Affinors\\ $L$ \end{tabular} \\\hline\\[-10pt]
(T3)
&
$
\begin{pmatrix}
 0 &-u^1 \\
 -u^1 & 0 \\
\end{pmatrix}
$ &
$
\begin{pmatrix}
 0 & \eta^{12} \\
 \eta^{12} & \eta^{22} \\
\end{pmatrix}
$
&
$
\begin{pmatrix}
-\frac{u^1}{\eta^{12}} &0 \\
\frac{\eta^{22}u^1}{(\eta^{12})^2} &  -\frac{u^1}{\eta^{12}} \\
\end{pmatrix}
$
\smallskip\\
(N5) 
&
$
\begin{pmatrix}
 0 &  u^1 \\
  u^1 & 2 (u^1+u^2) \\
\end{pmatrix}
$ &
$
\begin{pmatrix}
  0 &  \eta^{12}\\
  \eta^{12} & \eta^{22} \\
\end{pmatrix}
$
&
$
\begin{pmatrix}
\frac{u^1}{\eta^{12}} &0 \\
\frac{2 (u^1+u^2)}{\eta^{12}}-\frac{\eta^{22}u^1}{(\eta^{12})^2} &  \frac{u^1}{\eta^{12}} \\
\end{pmatrix}
$
\smallskip\\
(N3,N4,N6) 
&
$
\begin{array}{c}
\begin{pmatrix}
 0 &  (1+\kappa) u^1 \\
 (1+\kappa) u^1 & 2 u^2 \\
\end{pmatrix}\\
\end{array}
$ &
$
\begin{pmatrix}
  0 &  \eta^{12}\\
  \eta^{12} & \eta^{22} \\
\end{pmatrix}
$
&
$
\begin{pmatrix}
\frac{(1+\kappa)u^1}{\eta^{12}} &0 \\
\frac{2 u^2}{\eta^{12}}-\frac{(1+\kappa)\eta^{22}u^1}{(\eta^{12})^2} &  \frac{(1+\kappa)u^1}{\eta^{12}} \\
\end{pmatrix}
$
\smallskip\\

\hline
\end{longtable}
}

We prove that in the cases T3, N3 (corresponding to $\kappa=1$), N5 and N6 with  $\kappa\neq 0,-1, -2$ the deformations are quasi-trivial and can be reduced to the form
$$
\Pi_{\lambda}=\omega_{2}-\lambda\omega_{1}+\epsilon^2 {\rm Lie}_{X_{(F_1,F_2)}} \omega_{2}+\mathcal{O}(\epsilon^3)\\
$$
with $X_{(F_1,F_2)}=\omega_1\,\delta H-\omega_2\,\delta K$ where
$$
H[u]=\int \sum_{i,j} \left( h_{ij} u^i_x \log u^j_x  \right) \ dx,\qquad K[u]=\int \sum_{i,j} \left( f_{ij} u^i_x \log u^j_x  \right) \ dx,
$$
and the functions $h_{ij}$ and $f_{ij}$ are uniquely determined by two arbitrary functions $F_1,F_2$. Moreover both functions $F_1$ and $F_2$ depend only on the eigenvalue of the affinor $L$.

The cases N4  (corresponding to $\kappa=0$) and N6  with  $\kappa=-2$ are more involved and the functions labelling non Miura equivalent deformations are 4 (still depending on the eigenvalue of the affinor $L$). 

In all cases one half of the arbitrary functions parametrizing the deformations (one in the two-parameter case, two in the four-parameter case)  is related to a family of truncated structures and one half  is invariant with respect to the Miura transformations that preserve the dispersionless limit. The invariant functions are related to the first coefficients of the expansion \eqref{expansion}
 (in the second case  also the odd powers of $p$ appear in this expansion):
 the coefficients of $p^2$ in the case of the algebras T3, N3, N5 and N6 with  $\kappa\neq 0,-1, -2$ and the coefficients of $p$ and $p^2$ in the case of the algebras N4 and N6  with  $\kappa=-2$. Moreover our computations suggest that in the exceptional cases generic deformations are not quasi-trivial. This fact is rather unexepcted and deserves a deeper investigation.
\newline
\newline
\emph{The lift of deformations of semisimple structures}. These are obtained using an infinite dimensional version of the complete lift 
 introduced by Yano and Kobayashi in \cite{YK}. Whereas elementary, this case is important for it provides examples of full deformations of non semisimple structures depending on functional 
 parameters. By construction all deformations of a $n$-component semisimple structure can be lifted to deformations of a $2n$-component non semisimple structure.
 This means that the deformations of the lifted Poisson pencils contain $n$ functional parameters \emph{at least}.
This example suggests that also in the non semisimple case
 the deformations are unobstructed.

\section{Linear Poisson bivectors of hydrodynamic type}\label{sect_BN}
Let us introduce Poisson bivector of hydrodynamic type on the loop space $\mathcal L(M)$. 
The tangent space to $\mathcal L(M)$ at a loop $\gamma : S^1 \to M$ is naturally identified with the space $\Gamma(S^1,\gamma^*TM)$ of vector fields along $\gamma$.
On the other hand (a subspace of) the cotangent space to $\mathcal L(M)$ at $\gamma$ is identified with the space $\Gamma(S^1,\gamma^*T^*M)$ of covector fields along $\gamma$, and the pairing between a tangent vector $X$ and a covector $\xi$ is just $\int_{S^1} \xi(X)\,dx $.

Let $g$ be a pseudo-metric on $M$ with Levi-Civita connection $\nabla$. For any covector $\xi \in \Gamma(S^1,\gamma^*T^*M)$, let $X_\xi \in \Gamma(S^1,\gamma^*TM)$ be the pointwise metric dual of $\xi$.
Given two covectors $\xi,\eta \in \Gamma(S^1,\gamma^*T^*M)$, letting
\begin{equation*}
P(\xi,\eta)
= \int_{S^1} \xi(\nabla_{\dot \gamma} X_\eta)\,dx
\end{equation*}
defines a bivector on $\mathcal L(M)$.
As shown by Dubrovin and Novikov, $P$ is a Poisson structure on $\mathcal L(M)$ if and only if $\nabla$ is flat \cite{DN}.
In local coordinates $u^i$ on $M$ and $x$ on $S^1$ the Poisson tensor $P$ is represented by a differential operator of the form
\beq\label{pbht}
P^{ij}=g^{ij}(u)\d_x-g^{il}\Gamma^{j}_{lk}(u)u^k_x,
\eeq
where $\Gamma^{j}_{lk}$ are the  Christoffel symbols correponding to $g$.

Dubrovin-Novikov operators naturally appear in the study of Hamiltonian quasilinear systems of PDEs
$$
u^i_t=V^i_j(u)u^j_x,\qquad i=1,...,n,
$$
and their dispersive Hamiltonian deformations
$$
u^i_t=V^i_j(u)u^j_x+\epsilon\left(A^i_j(u)u^j_{xx}+B^i_{jk}(u)u^j_xu^k_x\right)+\mathcal{O}(\epsilon^2).
$$
In this paper we will study linear Hamiltonian operators. As proved by  Balinski\v{\i} and Novikov in \cite{BN} these operators
 have the form
$$
P^{ij}=(b^{ij}_k+b^{ji}_k)u^k\d_x+b^{ij}_{k}u^k_x,
$$
where the numbers $b^{ij}_k$ are the structure constants  of an algebra $B$ satisfying the following properties
\begin{gather*}
a \cdot ( b\cdot c)= b \cdot (a \cdot c), \label{lie1} \\
(a \cdot b) \cdot c-a \cdot( b \cdot c) = (a \cdot c) \cdot b-a \cdot( c \cdot b).  \label{lie2}
\end{gather*}
We refer to them as \emph{Balinski\v{\i}-Novikov algebras}, even if in the literature they are often called Novikov algebras (following \cite{O}).

A first approach to the study of such algebras was made by Zelmanov \cite{Z}. In low dimensions the problem of classification was addressed by Bai and Meng \cite{BM,BM1} and recently by Burde and de Graaf \cite{BdG}, resulting in a complete description of one-, two- and three-dimensional Balinski\v{\i}-Novikov algebras. Unfortunately, a full classification of these structures of dimension four and higher is far from being complete.

\subsection{Invariant bilinear forms and bi-Hamiltonian structures}
Given a Balinski\v{\i}-Novikov algebra $B$, as observed in \cite{SS}, any invariant bilinear symmetric form on it give rise to a bi-Hamiltonian structure in a canonical way.
For convenience of the reader let us briefly recall how they are defined.
Let $e^1, \ldots, e^n$ be a basis of $B$, and let $b^{ij}_{k}$ be the corresponding structure constants.
A bilinear form $\eta: B \times B \to F$ is called \emph{invariant} if and only if
\begin{equation*}\label{eq1}
\eta(e^i\cdot e^j, e^k)=\eta(e^i, e^k\cdot e^j).
\end{equation*}

Bai and Meng classified these invariant bilinear forms on two- and three-dimensional Balinski\v{\i}-Novikov algebras in \cite{BM,BMa}.
For future reference we recall the two-dimensional classification in the following table.

\LTcapwidth=\textwidth
\arraycolsep=4.5pt
\setlength{\tabcolsep}{3.5pt}
{\footnotesize
\begin{longtable}{ c  c  c  c }
\caption{Two-dimensional Balinski\v{\i}-Novikov algebras and invariant bilinear forms.}\label{tab1}\\
\hline
Type & \begin{tabular}{c} Characteristic \\matrix $e^i\cdot e^j$\end{tabular} &  \begin{tabular}{c}Linear Poisson \\structure \end{tabular} &  \begin{tabular}{c}Invariant \\bilinear forms \end{tabular} \\
\hline\\[-10pt]
(T1) & $
\begin{pmatrix}
 0 & 0 \\
 0 & 0 \\
\end{pmatrix}
$ &
 $
\begin{pmatrix}
 0 & 0 \\
 0 & 0 \\
\end{pmatrix}
$
&
$
\begin{pmatrix}
 \eta^{11} & \eta^{12} \\
 \eta^{21} & \eta^{22} \\
\end{pmatrix}
$
\smallskip\\
(T2) & $
\begin{pmatrix}
 e^2 & 0 \\
 0 & 0 \\
\end{pmatrix}
$ &
 $
\begin{pmatrix}
 2 u^2 \partial_x + u^2_x& 0 \\
 0 & 0 \\
\end{pmatrix}
$
&
$
\begin{pmatrix}
 \eta^{11} & \eta^{12} \\
 \eta^{12} & 0 \\
\end{pmatrix}
$
\smallskip\\
(T3) & $
\begin{pmatrix}
 0 & 0 \\
 -e^1 & 0 \\
\end{pmatrix}
$ &
$
\begin{pmatrix}
 0 &-u^1 \partial_x\\
 -u^1  \partial_x -u^1_x& 0 \\
\end{pmatrix}
$
&
$
\begin{pmatrix}
 0 & \eta^{12} \\
 \eta^{12} & \eta^{22} \\
\end{pmatrix}
$ 
\smallskip\\
(N1) & $
\begin{pmatrix}
 e^1 & 0 \\
 0 & e^2 \\
\end{pmatrix}
$ &
$
\begin{pmatrix}
 2 u^1 \partial_x  +u^1_x& 0 \\
 0 & 2 u^2\partial_x + u^2_x \\
\end{pmatrix}
$
&
$
\begin{pmatrix}
 \eta^{11} & 0 \\
 0 & \eta^{22} \\
\end{pmatrix}
$ 
\smallskip\\
(N2) & $
\begin{pmatrix}
 e^1 & 0 \\
 0 & 0 \\
\end{pmatrix}
$ &
$
\begin{pmatrix}
2 u^1 \partial_x + u^1_x & 0 \\
 0 & 0\\
\end{pmatrix}
$
&
$
\begin{pmatrix}
 \eta^{11} & 0 \\
 0 & \eta^{22} \\
\end{pmatrix}
$
\smallskip\\
(N3) & $
\begin{pmatrix}
 e^1 & e^2 \\
 e^2 & 0 \\
\end{pmatrix}
$ &
$
\begin{pmatrix}
 2 u^1 \partial_x + u^1_x& 2 u^2\partial_x + u^2_x \\
 2 u^2\partial_x + u^2_x & 0 \\
\end{pmatrix}
$
&
$
\begin{pmatrix}
 \eta^{11} & \eta^{12} \\
 \eta^{12} & 0 \\
\end{pmatrix}
$
\smallskip\\
(N4) & $
\begin{pmatrix}
 0 & e^1 \\
 0 & e^2 \\
\end{pmatrix}
$ &
$
\begin{pmatrix}
 0 & u^1 \partial_x + u^1_x\\
 u^1\partial_x  & 2 u^2 \partial_x + u^2_x\\
\end{pmatrix}
$
&
$
\begin{pmatrix}
 \eta^{11} & \eta^{12} \\
 \eta^{21} & \eta^{22} \\
\end{pmatrix}
$ 
\smallskip\\
(N5) & $
\begin{pmatrix}
 0 & e^1 \\
 0 & e^1+e^2 \\
\end{pmatrix}
$ &
$
\begin{pmatrix}
0 &  u^1 \partial_x + u^1_x \\
u^1 \partial_x  & 2 (u^1+u^2) \partial_x + u^2_x + u^1_x \\
\end{pmatrix}
$
&
$
\begin{pmatrix}
 0 &  \eta^{12} \\
  \eta^{12} & \eta^{22} \\
\end{pmatrix}
$
\smallskip\\
(N6) &
$
\begin{array}{c}
\begin{pmatrix}
 0 & e^1 \\
 \kappa e^1 & e^2 \\
\end{pmatrix}
\\
\kappa\neq0,1
\end{array}
$
&
$
\begin{pmatrix}
 0 &  (1+\kappa) u^1 \partial_x + u^1_x\\
 (1+\kappa) u^1\partial_x + \kappa u^1_x & 2 u^2 \partial_x + u^2_x \\
\end{pmatrix}
$
&
$
\begin{pmatrix}
0 & \eta^{12} \\
\eta^{12} & \eta^{22} \\
\end{pmatrix}
$ \smallskip\\
\hline
\end{longtable}
}

\begin{rmk}
Notice that  the case N4 with $\eta^{11}\ne 0$ is semisimple. For this reason we will consider only the case $\eta^{11}= 0$.
 The cases N3 and N4 can be considered as subcases of N6, if we remove the constraints $\kappa\neq0,1$. Indeed, for $\kappa=0$ we easily get N4 (with $\eta^{11}= 0$)  while N3 is equivalent to the case $\kappa=1$, up to swapping the local coordinates $u^1,u^2$. According to \cite{BM}, this distinction is due to different algebraic properties: the cases N3 and N4 are characterized by the associativity of the algebra, while this is not the case of N6 with $\kappa\neq0,1$. However, for our purposes, we do not need to distinguish these cases. 
\end{rmk}

Let us point out that adding the constraint $\eta^{21}=\eta^{12}$ in T1 and N4, the bilinear invariant forms associated with two-dimensional Balinski\v{\i}-Novikov algebra become symmetric. As observed by Strachan and Szablikowski in \cite{SS} the associated Hamiltonian operator $\eta^{ij}\d_x$ is compatible with the linear Hamiltonian operator defining the Balinski\v{\i}-Novikov algebra.

\begin{rmk}
A pair of compatible flat metrics defines a $(2+1)$-Poisson structure of hydrodynamic type under some additional conditions. Among the structures coming from two component Balinski\v{\i}-Novikov algebras, such additional conditions are satisfied just by $N6$ with $\kappa = -2$ \cite{DN1,M1,M,FLS}.
\end{rmk}

\subsection{Classification results}
In this section we provide a classification of second order deformations of Poisson pencils coming from Balinski\v{\i}-Novikov algebras.

By definition, a k-th deformation of a Poisson pencil of hydrodynamic type \eqref{eq::hdpencil} is a deformation \eqref{pencil} such that $[\tilde\Pi_\lambda,\tilde\Pi_\lambda]=\mathcal{O}(\epsilon^{k+1})$.
Here
where $\tilde \Pi_\lambda^{ij}$ denotes the distribution
\begin{equation*}
\begin{array}{l}
\displaystyle
\tilde\Pi^{ij}=\omega^{ij}_{2}+\sum_{k\ge 1}\epsilon^k\sum_{l=0}^{k+1}A^{ij}_{2;k,l}(u,u_x,\dots,u_{(l)})\delta^{(k-l+1)}(x-y)\\
\displaystyle
\qquad
-\lambda\left(\omega^{ij}_{1}+\sum_{k\ge 1}\epsilon^k\sum_{l=0}^{k+1}A^{ij}_{1;k,l}(u,u_x,\dots,u_{(l)})\delta^{(k-l+1)}(x-y)\right),
\end{array}
\end{equation*}
and the Schouten bracket is defined as follows \cite{DZ}:
\begin{multline*}
 [\tilde \Pi_\lambda,\tilde \Pi_\lambda]^{ijk}(x,y,z) =\\
2 \frac{\partial \tilde \Pi^{ij}_\lambda(x,y)}{\partial u^{l}_{(s)}(x)} \partial_x^s \tilde \Pi^{lk}_\lambda(x,z)
+ 2 \frac{\partial \tilde \Pi^{ki}_\lambda(z,x)}{\partial u^{l}_{(s)}(z)} \partial_z^s \tilde \Pi^{lj}_\lambda(z,y)
+ 2 \frac{\partial \tilde \Pi^{jk}_\lambda(y,z)}{\partial u^{l}_{(s)}(y)} \partial_y^s \tilde \Pi^{li}_\lambda(y,x),
\end{multline*}

We have to distinguish two cases:
\begin{enumerate}
\item The cases T3, N3, N5 and N6 with  $\kappa\neq 0,-1, -2$ where second order deformed structures depend on two functions.
\item The remaining cases N4 (which corresponds to $\kappa=0$) and N6 with $\kappa=-2$, namely
$$
g_1=
\begin{pmatrix}
0 & \eta^{12}\\
\eta^{12} & \eta^{22}
\end{pmatrix},
\quad
g_2=
\begin{pmatrix}
0 & \pm u^1\\
\pm u^1 & 2 u^2
\end{pmatrix},
$$
where second order deformed structures depend on four functions.
\end{enumerate}

\begin{thm}\label{thm_quasi}
In the cases T3, N3, N5 and N6 with $\kappa\neq 0,-1, -2$, second order deformations can be reduced by a Miura transformation to the form
$$
\Pi_{\lambda}=\omega_2-\lambda\omega_1+\epsilon^2 {\rm Lie}_{X_{(F_1,F_2)}} \omega_2+\mathcal{O}(\epsilon^3)\\
$$
with $X_{(F_1,F_2)}= \omega_1\,\delta H-\omega_2\delta K$ where
$$
H[u]=\int \sum_{i,j} \left( h_{ij} u^i_x \log u^j_x  \right) \ dx,\qquad K[u]=\int \sum_{i,j} \left( k_{ij} u^i_x \log u^j_x  \right) \ dx,
$$
and the functions $h_{ij}$ and $k_{ij}$ are uniquely determined in terms of two arbitrary functions $F_1,F_2$ depending only on the eigenvalue of the affinor $L=g_2g_1^{-1}$. Calling $\mathcal{K}=(k_{ij})$ and $\mathcal{H}=(h_{ij})$, we have $\mathcal{K}=L^T \mathcal{H}$, where $L^T$ means the transpose of $L$, and $\mathcal{H}$ is given respectively for each case by 
\begin{itemize}
\item T3: $h_{12}=h_{22}=0$ and
$$
h_{11}=
\frac{e^{-\frac{ \eta^{12} u^2}{ \eta^{22} u^1}}}{3  \eta^{12}}
\left(
\eta^{22} u^1 F_2' + \frac{\eta^{12} u^2 +   \eta^{22} u^1}{u^1} F_2
\right)
-F_1,
\quad
h_{21}=-\frac{e^{-\frac{ \eta^{12} u^2}{ \eta^{22} u^1}} }{3 } F_2.
$$
\item N5: $h_{12}=h_{22}=0$ and
$$
h_{11}=\frac{\sqrt{2 \eta^{12} (u^1+u^2) - \eta^{22} u^1} F_2'}{3 \eta^{12}}
+\frac{(2 \eta^{12} - \eta^{22})F_2}{6 \eta^{12} \sqrt{2 \eta^{12} (u^1+u^2) - \eta^{22} u^1}}
+\frac{ F_1}{2 \eta^{12}},
$$
$$
h_{21}=\frac{1}{3 \sqrt{2 \eta^{12} (u^1+u^2) - \eta^{22} u^1}} F_2.
$$
\item N3, N6 $(\kappa\neq 0,-1, -2)$: $h_{12}=h_{22}=0$ and
\begin{gather*}
h_{11}=
\frac{ (2 \eta^{12} u^2 -(\kappa+1) \eta^{22} u^1)^{\frac{\kappa+1}{2}} F_2'}{3 (\kappa+1)^2 \eta^{12}}
-
\frac{\eta^{22}  (2 \eta^{12} u^2 -(\kappa+1) \eta^{22} u^1)^{\frac{\kappa-1}{2}} F_2}{6 \eta^{12}}\\
\quad
+
\frac{F_1}{\eta^{12} \kappa (\kappa+2)},
\end{gather*}
$$
h_{21}=\frac{(2 \eta^{12} u^2 -(\kappa+1) \eta^{22} u^1)^{\frac{k-1}{2}}}{3 (\kappa+1)} F_2.
$$
\end{itemize}
Here $F_i=F_i(u^1)$, $i=1,2$.
\newline
\newline
In the case N4, namely
$$ 
g_2=
\begin{pmatrix}
0 & \eta^{12}\\
\eta^{12} & \eta^{22}
\end{pmatrix},
\quad
g_1=
\begin{pmatrix}
0 &  u^1\\
 u^1 & 2 u^2
\end{pmatrix},
$$
the second order deformations can be reduced by a Miura transformation to the form
$$
\Pi_{\lambda}=\omega_2-\lambda\omega_1+\epsilon^2 {\rm Lie}_{X} \omega_2+\mathcal{O}(\epsilon^3)\\
$$
where 
$$
X^i=X^i_1u^1_{xx} + X^i_2 (u^1_x)^2+X^i_3 u^1_x u^2_x + X^i_4(u^2_x)^2+X^i_5 u^2_{xx},  
$$
with
\begin{eqnarray*}
X^1_1&=& 0,\\
X^1_2&=& \theta F_1, \\
X^1_3&=& \partial_1 (\theta F_2) \\
X^1_4&=& \partial_2 (\theta F_2), \\
X^1_5&=& \theta F_2, \\
X^2_1&=& 0,\\
X^2_2&=& \theta F_3, \\
X^2_3&=&  \partial_1\left(\theta^{\frac{1}{2}} F_4-\frac{\partial_1 F_2}{\eta^{12} }\right),\\
X^2_4&=&  \partial_2\left(\theta^{\frac{1}{2}} F_4-\frac{\partial_1 F_2}{\eta^{12} }\right),\\
X^2_5&=&  \theta^{\frac{1}{2}} F_4- \frac{\partial_1 F_2}{\eta^{12} }.
\end{eqnarray*}
In the above formulas $F_i$ are $4$ arbitrary functions of $u^1$ and $\theta=(\eta^{22}  u^1-2 \eta^{12}  u^2)^{-1}$.
\newline
\newline
In the case N6 with $\kappa=-2$, namely
$$
g_1=
\begin{pmatrix}
0 & \eta^{12}\\
\eta^{12} & \eta^{22}
\end{pmatrix},
\quad
g_2=
\begin{pmatrix}
0 & - u^1\\
- u^1 & 2 u^2
\end{pmatrix},
$$
\newline
\newline
the second order deformations can be reduced by a Miura transformation to the form
\beq
\Pi_{\lambda}=\omega_2-\lambda\omega_1+\epsilon^2 {\rm Lie}_{X} \omega_2+\mathcal{O}(\epsilon^3)\\
\eeq
where
$$
X^i=X^i_1u^1_{xx} + X^i_2 (u^1_x)^2+X^i_3 u^1_x u^2_x + X^i_4(u^2_x)^2+X^i_5 u^2_{xx}, 
$$
with
\begin{eqnarray*}
X^1_1&=& 0,\\
X^1_2&=& 2 \eta^{22} \theta \left(\theta^{\frac{3}{2}} F_4 -\frac{\partial_1 (\theta^2 F_2)}{\eta^{12}} \right)+ \theta F_1,\\
X^1_3&=& 2  \eta^{12} \theta^{\frac{5}{2}} F_4 - \partial_1 (\theta^3 F_2),\\
X^1_4&=& -4  \eta^{12} \theta^4 F_2,\\
X^1_5&=& \theta^3 F_2,\\
X^2_1&=& 0,\\
X^2_2&=& F_3,\\
X^2_3&=&\partial_1 ( \theta^{\frac{3}{2}} F_4)-\frac{\partial_1^2(\theta^2 F_2)}{\eta^{12}},\\
X^2_4&=& 4 \partial_1(\theta^3 F_2)+ \partial_2 ( \theta^{\frac{3}{2}} F_4),\\
X^2_5&=& \theta^{\frac{3}{2}} F_4 -\frac{ \partial_1(\theta^2 F_2)}{\eta^{12}}.
\end{eqnarray*}
In the above formulas $F_i$ are $4$ arbitrary functions of $u^1$ and $\theta=(2 \eta^{12}  u^2 + \eta^{22}  u^1)^{-1}$.
\end{thm}

Due to its technical nature, we postpone the proof to Appendix \ref{appA}.

\begin{coro}
In the cases T3, N3, N5 and N6 with  $\kappa\neq 0,-1, -2$, all second order deformations are quasi-trivial.
\end{coro}

\proof
By construction the canonical quasi-Miura transformation generated by $H[u]$ reduces the pencil to its dispersionless limit up to terms of order $\mathcal{O}(\epsilon^4)$.

\endproof

\begin{rmk}
General Miura transformations have the form
\begin{equation*}\label{Miura_gen}
u^i \to \tilde{u}^i=f^i(u) + \sum_{k\ge1} \epsilon^k F^i_k(u,u_x,\ldots,u_{(k)}).
\end{equation*}
where ${\rm det}\frac{\d f^i}{\d u^j}\ne 0$. In this paper we are interested in  Miura transformations preserving the 
 disperionless limit and for this reason we consider the subgroup
\begin{equation*}\label{Miura_kind2}
u^i \to \tilde{u}^i=u^i + \sum_{k\ge1} \epsilon^k F^i_k(u,u_x,\ldots,u_{(k)}).
\end{equation*}
Indeed, the only diffeomorphism preserving both metrics of the pencil is the identity.
\end{rmk}

\subsection{Invariants of bi-Hamiltonian structures}
As already mentioned in the Introduction, the central invariants for deformations of  semisimple Poisson pencils of hydrodynamic type
 \eqref{pencil} are related to the roots of the equation
$$
{\rm det}\left(g^{ij}_{2}-\lambda g^{ij}_{1}+\sum_{k\ge 1}\left(A^{ij}_{2;k,0}(u)-\lambda A^{ij}_{1;k,0}(u)\right)p^{k}\right)=0.
$$
Expanding these roots  near $\lambda^i=r^i$ one obtains a series:
\begin{equation}\label{expansion2}
\lambda^i=r^i+\sum_{k=1}^{\infty}\lambda^i_{k} p^{k},
\end{equation}
whose coefficients are invariants (up to permutations) with respect to Miura transformations as shown by Dubrovin, Liu and Zhang in \cite{DLZ2}.

Due to the skew-symmetry of the pencil, the sum and product of the roots contain only even powers of $p$. In the semisimple case   also the expansions \eqref{expansion2} of the roots contain only even powers of $p$, while in the non semisimple case, in general
 also odd powers are allowed. For instance, in the case of deformations of non semisimple pencils associated with {Balinski\v{\i}-Novikov algebras one obtains the expansions 
\begin{eqnarray}\label{expansion3}
\lambda^1=u^1+\sum_{k=1}^{\infty}\lambda^1_{k} p^{k},\qquad\lambda^2=u^1+\sum_{k=1}^{\infty}\lambda^2_{k} p^{k}.
\end{eqnarray}
where, due to skew-symmetry:
\begin{equation}\label{sks}
\lambda^1_{2k+1}+\lambda^2_{2k+1}=0,\qquad \lambda^1_{2k}-\lambda^2_{2k}=0.
\end{equation}
Thus it is natural to divide Poisson  pencils associated with {Balinski\v{\i}-Novikov algebras in two classes: those admitting 
 as invariants $\lambda^1_{1}=-\lambda^2_{1}$ and  $\lambda^1_{2}=\lambda^2_{2}$  (and eventually higher order coefficients of the expansions \eqref{expansion3}) and those admitting as invariants
 only $\lambda^1_{2}-\lambda^2_{2}$ (and eventually higher order coefficients of the expansions \eqref{expansion3}).

\subsubsection{The cases T3, N3, N5 and N6 with  $\kappa\neq 0,-1, -2$.}
In the cases T3, N3, N5 and N6 with  $\kappa\neq 0,-1, -2$, the expansions of $\lambda^i$ do not contain
 the linear term in $p$ and the coefficients of the quadratic terms $\lambda^1_2=\lambda^2_2$ are related to the functional parameter $F_2$.

\begin{thm}\label{thm_main}
Let $\omega_{\lambda}=\omega_2-\lambda \omega_1$ bi-Hamiltonian structure corresponding to one of the Balinski\v{\i}-Novikov algebras T3, N3, N5 and N6 with  $\kappa\neq 0,-1, -2$ and the associated symmetric bilinear invariant form $\eta$. Let us consider a bi-Hamiltonian structures $\Pi_{\lambda}$ of the form \eqref{pencil} with leading term $\omega^{ij}_\lambda$. Then
 the coefficients $\lambda^1_2$ and $\lambda^2_2$ of the expansion \eqref{expansion2} coincide and they are related to the functional parameter $F_2$ by the formulas:
\begin{itemize}
\item T3: $\displaystyle\lambda^i_2=\frac{u^1}{\eta^{12}}e^{-\frac{\eta^{12} u^2}{\eta^{22} u^1}}F_2 (u^1)$.
\item N5: $\displaystyle\lambda^i_2= -\frac{u^1 F_2(u^1)}{\eta^{12} \sqrt{2 \eta^{12}(u^1+u^2)-\eta^{22} u^1}}$.
\item N3, N6 with  $\kappa\neq 0,-1, -2$: 
$\displaystyle\lambda^i_2= - \frac{(\kappa+1) u^1 (2 \eta^{12} u^2 -(\kappa+1) \eta^{22} u^1)^{\frac{\kappa-1}{2}} }{\eta^{12}} F_2(u^1).$
\end{itemize}
\end{thm}

\proof
We are going to prove this statement in the case T3 with $\eta^{22}\neq0$. In this case the dispersionless limit is given by
$$
\omega^{ij}_1=
\begin{pmatrix}
0 & \eta^{12}\\
\eta^{12} & \eta^{22} 
\end{pmatrix}
\partial_x,
\quad
\omega^{ij}_2=
\begin{pmatrix}
0 & - u^1\\
- u^1 & 0
\end{pmatrix}
\partial_x
+
\begin{pmatrix}
0 & 0\\
-u^1_x& 0
\end{pmatrix}.
$$
If we write the pencil in the standard form
\begin{eqnarray*}
\label{generic}
\Pi^{ij}_{\lambda}
&=&\omega^{ij}_{\lambda}+\sum_{k=1}^2\epsilon^k\sum_{l=0}^{k+1}\left(A^{ij}_{2;k,l}(u,\dots,u_{(l)})-\lambda A^{ij}_{1;k,l}(u,\dots,u_{(l)})\right)\d_x^{(k-l+1)}+ \mathcal{O}(\epsilon^3)
\end{eqnarray*}
the first two terms of the expansion \eqref{expansion} are
\begin{eqnarray}
\label{inv1}
\lambda^i_2&=&0\\
\label{inv2}
\lambda^i_2&=&\frac{1}{\eta^{12}}\left(Q^{12}_2+\frac{(P^{12}_2)^2}{u^1} + \frac{\eta^{22} Q^{11}_2}{2 \eta^{12}}+ \frac{u^1 Q^{12}_1 + P^{12}_1 P^{12}_2}{\eta^{12}}\right),
\end{eqnarray}
where
\begin{equation*}\label{PQ}
P^{ij}_{\theta}(u)=A^{ij}_{\theta;1,2}(u), \quad Q^{ij}_{\theta}(u)=A^{ij}_{\theta;2,3}(u),\quad i,j=1,\ldots, n, \quad \theta=1,2.
\end{equation*}
We know from general theory that these coefficients are invariant up to permutations. The condition $\lambda^1_{2n}=\lambda^2_{2n}$ implies that are genuine invariants.

Using this the proof is a straightforward computation: substituting the relations
\begin{equation*}\label{N3_Omega}
P_1=P_2=Q_1=\begin{pmatrix}
0 & 0\\
0 & 0
\end{pmatrix}
,\qquad Q_2=
\begin{pmatrix}
0 & u^1 e^{-\frac{\eta^{12} u^2}{\eta^{22} u^1}} F_2 (u^1)\\
u^1 e^{-\frac{\eta^{12} u^2}{\eta^{22} u^1}} F_2 (u^1) & *
\end{pmatrix},
\end{equation*}
in the formula \eqref{inv2} we get the result. Remaining cases can be proved following the same procedure. 

\endproof

\begin{rmk}
The invariant $\lambda^i_2$ can be also written as
$$\lambda^i_2=-\frac{1}{2} \mathrm{Res}_{\lambda=\hat{\lambda}} \mathrm{Tr} (g_{\lambda}^{-1} \Lambda_{\lambda})$$
where $\hat{\lambda}$ is the eigenvalue of the affinor $L=g_{2} g_{1}^{-1}$ and $\Lambda^{ij}_{\lambda}=Q^{ij}_{\lambda}+\frac{1}{2}(g_{\lambda}^{-1})_{lk} P^{li}_\lambda P^{kj}_{\lambda}.$ 
\end{rmk}

\subsubsection{The cases N4  and N6 with $\kappa= -2$}
In the remaining cases  the expansion of $\lambda^i$ contains also
 the linear term in $p$ and the invariants $\lambda^1_1=-\lambda^2_1$ and  $\lambda^1_2=\lambda^2_2$ are related  to the functional parameters $F_2$ and $F_4$ respectively.

\begin{thm}\label{thm_2}
Let $\omega_{\lambda}=\omega_2-\lambda \omega_1$ bi-Hamiltonian structure corresponding to one of the Balinski\v{\i}-Novikov algebras  N4  and N6 with $\kappa= -2$ and the associated symmetric bilinear invariant form $\eta$. Let us consider a bi-Hamiltonian structures $\Pi_{\lambda}$ of the form \eqref{pencil} with leading term $\omega^{ij}_\lambda$. Then, the invariants $(\lambda^i_1)^2$ and $\lambda^i_2$  are related to the functional parameters $F_2$ and $F_4$ through the formulas:
\begin{itemize}
\item N4:
\begin{eqnarray*}
(\lambda^i_1)^2&=&\frac{2 u^1 F_2}{(\eta^{12})^3},\\
\lambda^i_2&=&\frac{\partial_1( u^1 F_2)}{(\eta^{12})^2}-\frac{u^1 F_4}{\eta^{12} \sqrt{-2 \eta^{12} u^2 +\eta^{22} u^1}}.
\end{eqnarray*}
\item N6, $\kappa=-2$:
\begin{eqnarray*}
(\lambda^i_1)^2&=&\frac{2 u^1 F_2}{(\eta^{12})^3 (2 \eta^{12} u^2 +\eta^{22} u^1)^{2}},\\
\lambda^i_2&=&\frac{u^1 F_4}{\eta^{12} (2 \eta^{12} u^2 +\eta^{22} u^1)^{3/2}}
-\frac{(2 \eta^{12} u^2 -\eta^{22} u^1) F_2+u^1 F_2'}{ (\eta^{12})^2 (2 \eta^{12} u^2 +\eta^{22} u^1)^{3}}.
\end{eqnarray*}
\end{itemize}
\end{thm}

\proof
We outline the proof in the case N4 (corresponding to $\kappa=0$). In this case, the standard form of the pencil is
$$
\tilde\Pi^{ij}_{\lambda}=\omega^{ij}_\lambda +  \epsilon^2 \Theta^{ij}+ \mathcal{O}(\epsilon^3)=
\omega^{ij}_\lambda +  \epsilon^2\left(\Theta^{ij}_{(3)}\d_x^3+\Theta^{ij}_{(2)}\d_x^2+\Theta^{ij}_{(1)}\d_x+\Theta^{ij}_{(0)}\right)+ \mathcal{O}(\epsilon^3),
$$
where
$$
\omega^{ij}_\lambda=\begin{pmatrix}
0 & u^1\\
u^1 & 2 u^2
\end{pmatrix}
\d_x
+
\begin{pmatrix}
0 & u^1_x\\
0 &  u^2_x
\end{pmatrix}
-\lambda
\begin{pmatrix}
0 & \eta^{12}\\
\eta^{12} & \eta^{22} 
\end{pmatrix}
\d_x.
$$
and
\begin{equation*}
\Theta_{(3)}=
\begin{pmatrix}
\frac{2 u^1 F_2}{2 \eta^{12} u^2 -\eta^{22} u^1} & \frac{u^1 F_2'}{\eta^{12}} -\frac{u^1 F_4}{\sqrt{-2 \eta^{12} u^2 +\eta^{22} u^1}}+\frac{2 u^2 F_2}{2 \eta^{12} u^2 -\eta^{22} u^1}\\
\frac{u^1 F_2'}{\eta^{12}} -\frac{u^1 F_4}{\sqrt{-2 \eta^{12} u^2 +\eta^{22} u^1}}+\frac{2 u^2 F_2}{2 \eta^{12} u^2 -\eta^{22} u^1} & \frac{4u^2 F_2'}{\eta^{12}} -\frac{4 u^2 F_4}{\sqrt{-2 \eta^{12} u^2 +\eta^{22} u^1}}
\end{pmatrix},
\end{equation*}
From the general theory and from relations \eqref{sks} we know that $(\lambda^i_1)^2$ and $\lambda^i_2$ are invariants. Using the invariance the proof is a straightforward computation. The case N6 with $\kappa=-2$ can be treated in a similar way.
\endproof

\begin{rmk}
The function $\Theta^{12}_{(3)}$ can be also written as
$$\Theta^{12}_{(3)}=-\frac{\eta^{12}}{2} \mathrm{Res}_{\lambda=\hat{\lambda}} \mathrm{Tr} (g_{\lambda}^{-1} \Lambda_{\lambda})$$
where $\hat{\lambda}$ is the eigenvalue of the affinor $L=g_{2} g_{1}^{-1}$ and $\Lambda^{ij}_{\lambda}=Q^{ij}_{\lambda}+\frac{1}{2}(g_{\lambda}^{-1})_{lk} P^{li}_\lambda P^{kj}_{\lambda}.$ 
\end{rmk}

\section{Truncated structures}
In Theorems \ref{thm_main}, \ref{thm_2} we proved the invariant nature of some functional parameters appearing in deformations. In this section we prove that the remaining parameters are related to truncated structures.  These are Poisson pencils of the form \eqref{pencil} depending polynomially on the parameter $\epsilon$ (that is the sum in \eqref{pencil} contains finitely many terms). We show that setting to zero the invariant parameters we obtain deformations that 
are Miura equivalent to truncated pencils up to the order three. More precisely we prove that in the cases T3, N3, N5 and N6 with  $\kappa\neq 0,-1, -2$ the additional parameter provides a one-parameter family of truncated structures, while in the cases N4 and N6 with $\kappa=-2$ the two additional parameters provide a two-parameter family of truncated structures.  
  

\begin{thm}
In the cases T3, N3, N5 and N6 with  $\kappa\neq 0,-1, -2$, the second order deformations with  $F_2=0$ can be reduced by a Miura transformation to the form $\Pi_{\lambda}=\omega_{\lambda}+\epsilon^2 \Theta+\mathcal{O}(\epsilon^3)$ where
\begin{equation}\label{truncated1}
\Theta=
\begin{pmatrix}
0 & 0\\
0 & 2f
\end{pmatrix}
\d_x^3
+
\begin{pmatrix}
0 & 0\\
0 & 3 f_x
\end{pmatrix}
\d_x^2
+
\begin{pmatrix}
0 & 0\\
0 & f_{xx}
\end{pmatrix}
\d_x,
\end{equation}
with $f=f(u^1)$. Moreover  the truncated pencil $\omega_{\lambda}+\epsilon^2 \Theta$ is a Poisson pencil.
\end{thm}
\proof
The form \eqref{truncated1} can be easily obtained from the results of Theorem \ref{thm_quasi} rescaling the function $F_1$. In particular, we have to set
\begin{itemize}
\item $F_1(u^1)=\dfrac{f(u^1)}{u^1}$, for T3,
\item $F_1(u^1)=-\dfrac{\eta^{12}f(u^1)}{u^1}$, for N5,
\item $F_1(u^1)=-\dfrac{\eta^{12} \kappa f(u^1)}{(1+\kappa) u^1}$, for N3, N6 with  $\kappa\neq 0,-1, -2$.
\end{itemize}

To prove that $\omega_{\lambda}+\epsilon^2 \Theta$ is a Poisson pencil, we have to show that
\begin{multline*}
 \frac{1}{2}[\Theta,\Theta]^{ijk}(x,y,z) = \\
\frac{\partial \Theta^{ij}(x,y)}{\partial u^{l}_{(s)}(x)} \partial_x^s \Theta^{lk}(x,z)
+ \frac{\partial \Theta^{ki}(z,x)}{\partial u^{l}_{(s)}(z)} \partial_z^s \Theta^{lj}(z,y)
+ \frac{\partial \Theta^{jk}(y,z)}{\partial u^{l}_{(s)}(y)} \partial_y^s \Theta^{li}(y,x)
=0.
\end{multline*}
Taking into account that $\Theta^{11}=\Theta^{12}=\Theta^{21}=0$ and $\frac{\partial \Theta^{22}}{\partial u^{2}_{(s)}}=0,$  
we obtain the result.

\endproof

\begin{thm}
In the case N6 with  $\kappa= -2$ the second order deformations with $F_2=F_4=0$ can be reduced by a Miura transformation to the form $\Pi_{\lambda}=\omega_{\lambda}+\epsilon^2 \Theta+\mathcal{O}(\epsilon^3)$ where
\begin{equation}\label{truncated2}
\Theta=
\begin{pmatrix}
0 & 0\\
0 & 2f
\end{pmatrix}
\d_x^3
+
\begin{pmatrix}
0 & 0\\
0 & 3 f_x 
\end{pmatrix}
\d_x^2
+
\begin{pmatrix}
0 & 0\\
0 & f_{xx} + 2 g
\end{pmatrix}
\d_x
+
\begin{pmatrix}
0 & 0\\
0 & g_x
\end{pmatrix},
\end{equation}
with $f=f(u^1)$ and $g=\left(h(u^1)  u^1_x\right)_x + h(u^1) u^1_{xx}$. Moreover the truncated pencil $\omega_{\lambda}+\epsilon^2 \Theta$
 is a Poisson pencil.
\end{thm}

\proof
Here we prove only the first part of the theorem. The second part can be obtained as above by straightforward computation.

By Theorem \ref{thm_quasi} we have 
$$\Pi_{\lambda}=\omega_2-\lambda \omega_1+\epsilon^2 {\rm Lie}_X \omega_2+\mathcal{O}(\epsilon^3),$$ 
where the component of the vector field $X$ are given by
$$
X^1= \theta F_1 (u^1_x)^2, \quad X^2= F_3 (u^1_x)^2,
$$
with $\theta=(2 \eta^{12} u^2 + \eta^{22} u^1)^{-1}$.
 The Miura transformation
$$
u^i \to \exp(-\epsilon Y) u^i, \quad i=1,2,
$$
 generated by the vector field $Y$ of components
\begin{eqnarray*}
Y^1&=&-\eta^{12} R u^1_{xx}-\eta^{12} \partial_1  R (u^1_x)^2-\eta^{12} \partial_2 R  u^1_x u^2_x,\\
Y^2&=& -\eta^{22} R  u^1_{xx} -\eta^{22} \partial_1  R (u^1_x)^2 + (\eta^{12} \partial_1  R -\eta^{22} \partial_2  R ) u^1_x u^2_x + \eta^{12} \partial_2  R (u^2_x)^2 + \eta^{12} R u^2_{xx},
\end{eqnarray*}
with $R=\frac{u^1 F_1}{2\eta^{12}(2 \eta^{12} u^2 + \eta^{22} u^1) }$, reduces the pencil to the form  $\omega_2-\lambda \omega_1+\epsilon^2 {\rm Lie}_{\tilde X} \omega_2+\mathcal{O}(\epsilon^3)$, where
\begin{eqnarray*}
\tilde X^1&=&-\frac{\theta u^1 F_1 u^1_{xx}}{2}
-\left(\frac{\theta u^1 F_1'}{2} - \theta^2 (\eta^{12} u^2 +\eta^{22} u^1)F_1\right) (u^1_x)^2
+ \theta^2 \eta^{12} u^1 F_1 u^1_x u^2_x\\
\tilde X^2&=&-\frac{\theta \eta^{22} u^1 F_1 u^1_{xx}}{2 \eta^{12}}+ \frac{\theta u^1 F_1 u^2_{xx}}{2}
+ \left( \frac{\theta u^1 F_1'}{2} + \theta^2 (\eta^{12} u^2 +\eta^{22} u^2) F_1\right) u^1_x u^2_x\\
&&-\left(\frac{\theta \eta^{22} u^1 F_1'}{2 \eta^{12}} + \theta^2 \eta^{22} u^2 F_1 - F_3\right) (u^1_x)^2
- \theta^2 \eta^{12} u^1 F_1 (u^2_x)^2.
\end{eqnarray*}
To conclude it is easy to check that ${\rm Lie}_{\tilde X} \omega_2$ coincides with \eqref{truncated2} ($F_1=-\frac{2 \eta^{12} f }{u^1}$ and $F_3=-\frac{h }{u^1}$).
\endproof

\begin{thm}
In the case N4 with $F_2=F_4=0$ the second order deformations can be reduced by a Miura transformation to the form $\Pi_{\lambda}=\omega_{\lambda}+\epsilon^2 \Theta+\mathcal{O}(\epsilon^3)$ where
\begin{equation}\label{truncated3}
\Theta=
\begin{pmatrix}
0 & 0\\
0 & q^{22}_3
\end{pmatrix}
\d_x^3
+
\begin{pmatrix}
0 & q^{12}_2\\
-q^{12}_2 & q^{22}_2 
\end{pmatrix}
\d_x^2
+
\begin{pmatrix}
q^{11}_1 & q^{12}_1\\
q^{21}_1 & q^{22}_1
\end{pmatrix}
\d_x
+
\begin{pmatrix}
q^{11}_0 & q^{12}_0\\
q^{21}_0 & q^{22}_0
\end{pmatrix},
\end{equation}
with
\begin{eqnarray*}
q^{22}_3&=& 2 f,\\
q^{12}_2&=& 4 \theta \eta^{12} f u^1_x,\\
q^{22}_2&=& 3 f' u^1_x, \\
q^{11}_1&=& -8 (\theta \eta^{12})^2 f (u^1_x)^2,\\
q^{12}_1&=& (2 \theta \eta^{12} f' -2\theta^2 \eta^{12} \eta^{22} f+2 \theta^2 h) (u^1_x)^2,\\
q^{21}_1&=& (-6 \theta \eta^{12} f' -10\theta^2 \eta^{12} \eta^{22} f+2 \theta^2 h) (u^1_x)^2
+16 (\theta \eta^{12})^2 f u^1_x u^2_x -8  \theta \eta^{12} f u^1_{xx},\\
q^{22}_1&=& (f'' + 2 \theta (\eta^{12})^{-1}  h' +6 \theta^2 (\eta^{12})^{-1}\eta^{22}  h) (u^1_x)^2
-8 \theta^2 h u^1_x u^2_x +(f' + 4 \theta(\eta^{12})^{-1} h) u^1_{xx},\\
q^{11}_0&=& -\left( 4 (\theta \eta^{12})^2   f' + 8 \theta^3(\eta^{12})^2 \eta^{22} f   \right) (u^1_x)^3
+ 16 (\theta \eta^{12})^3 f (u^1_x)^2 u^2_x - 8 (\theta \eta^{12})^2 f u^1_x u^1_{xx},\\
q^{12}_0&=&  (2 \theta^2 h' + 4 \theta^3 \eta^{22} h) (u^1_x)^3 - 8 \theta^3\eta^{12} h (u^1_x)^2 u^2_x +4 \theta^2 h u^1_x u^1_{xx} ,\\
q^{21}_0&=&
( - 2 \theta \eta^{12} f'' - 8 \theta^2 \eta^{12} \eta^{22} f' -12 \theta^3 \eta^{12} (\eta^{22})^2 f) (u^1_x)^3\\
&&
+(12 (\theta \eta^{12})^2 f' + 40 \theta^3 (\eta^{12})^2 \eta^{22} f) (u^1_x)^2 u^2_x
+(- 8 \theta \eta^{12} f' - 16 \theta^2 \eta^{12} \eta^{22}  f    ) u^1_x u^1_{xx}\\
&&
-32 (\theta \eta^{12})^3 f u^1_x (u^2_x)^2
+8 (\theta \eta^{12})^2 f u^1_x u^2_{xx}
+16 (\theta \eta^{12})^2 f u^1_{xx} u^2_x
- 4 \theta \eta^{12} f u^1_{xxx}
,\\
q^{22}_0&=&
( \theta (\eta^{12})^{-1} h'' + 4  \theta^2  (\eta^{12})^{-1} \eta^{22} h' + 6  \theta^3  (\eta^{12})^{-1} (\eta^{22})^2 h ) (u^1_x)^3\\
&&
+(- 6 \theta^2 h- 20 \theta^3 \eta^{22} h ) (u^1_x)^2 u^2_x
+(  4  \theta  (\eta^{12})^{-1}  h'  + 8  \theta^2  (\eta^{12})^{-1} \eta^{22} h ) u^1_x u^1_{xx}\\
&&
+16 \theta^3 \eta^{12} h u^1_x (u^2_x)^2
-2 \theta^2 h u^1_x u^2_{xx}
-4 \theta^2 h u^1_{xx} u^2_x
+ \theta (\eta^{12})^{-1} h u^1_{xxx}
,
\end{eqnarray*}
where $f=f(u^1)$, $h=h(u^1)$ and $\theta=(2 \eta^{12} u^2 - \eta^{22} u^1)^{-1}$. Moreover the truncated pencil $\omega_{\lambda}+\epsilon^2 \Theta$ is a Poisson pencil.
\end{thm}

\proof
By Theorem \ref{thm_quasi} we have $\Pi_{\lambda}=\omega_2-\lambda \omega_1+\epsilon^2 {\rm Lie}_X \omega_2+\mathcal{O}(\epsilon^3)$, where the components of the vector field $X$ are given by
$$
X^1= -\theta F_1 (u^1_x)^2, \quad X^2= -\theta F_3 (u^1_x)^2,
$$
with $\theta=(2 \eta^{12} u^2 - \eta^{22} u^1)^{-1}$.
 The Miura transformation 
$$
u^i \to \exp(-\epsilon Y) u^i, \quad i=1,2,
$$
generated by the vector field $Y$ of components
\begin{eqnarray*}
Y^1&=&-\eta^{12} R u^1_{xx}-\eta^{12} \partial_1  R (u^1_x)^2-\eta^{12} \partial_2 R  u^1_x u^2_x,\\
Y^2&=& -\eta^{22} R  u^1_{xx} -\eta^{22} \partial_1  R (u^1_x)^2 + (\eta^{12} \partial_1  R -\eta^{22} \partial_2  R ) u^1_x u^2_x + \eta^{12} \partial_2  R (u^2_x)^2 + \eta^{12} R u^2_{xx},
\end{eqnarray*}
with $R=-\frac{u^1 F_1}{2\eta^{12}(2 \eta^{12} u^2 - \eta^{22} u^1) }$, reduces the pencil to the form 
$$\omega_2-\lambda \omega_1+\epsilon^2 {\rm Lie}_{\tilde X} \omega_2+\mathcal{O}(\epsilon^3),$$ 
where
\begin{eqnarray*}
X^1&=&\frac{\theta u^1 F_1 u^1_{xx}}{2} + \left( \frac{\theta u^1 F_1'}{2}- \theta^2 (\eta^{12} u^2 -\eta^{22} u^2) F_1\right) (u^1_x)^2 -\theta^2\eta^{12} u^1 F_1 u^1_x u^2_x,\\
X^2&=&\frac{\theta \eta^{22} u^1 F_1 u^1_{xx}}{2 \eta^{12}}- \frac{\theta u^1 F_1 u^2_{xx}}{2}
- \left( \frac{\theta u^1 F_1'}{2} + \theta^2 (\eta^{12} u^2 +\eta^{22} u^2) F_1\right) u^1_x u^2_x\\
&&+\left(\frac{\theta \eta^{22} u^1 F_1'}{2 \eta^{12}} + \theta^2 \eta^{22} u^2 F_1 -\theta F_3\right) (u^1_x)^2
+ \theta^2 \eta^{12} u^1 F_1 (u^2_x)^2,
\end{eqnarray*}
To conclude the first part of the theorem we observe that it is easy to check that ${\rm Lie}_{\tilde X} \omega_2=\Theta$ 
 ($F_1=\frac{2 \eta^{12} f }{u^1}$ and $F_3=-\frac{h }{\eta^{12} u^1}$). 
The second part is a cumbersome computation.

\endproof

\begin{rmk}
Truncated Poisson pencils of the form
\begin{equation}\label{SS}
\Pi^{ij}_{\lambda}=\omega^{ij}_{\lambda}+\epsilon\sum_{l=0}^{2}(A^{ij}_{2;1,l}-\lambda A^{ij}_{1;1,l})
\d_x^{(2-l)}+\epsilon^2\sum_{l=0}^{3}(A^{ij}_{2;2,l}-\lambda A^{ij}_{1;2,l})\d_x^{(3-l)}\\
\end{equation}
where $\omega_\lambda$ is a Poisson pencil of hydrodynamic type associated with a Balinski\v{\i}-Novikov algebra appear in  \cite{SS}. In this case the coefficients 
$$A^{ij}_{2;1,0},\,A^{ij}_{1;1,0},\,A^{ij}_{2;2,0},\,A^{ij}_{1;2,0}$$ 
are related   with second and third order cocycles of the Balinski\v{\i}-Novikov algebra. In order to reduce deformations of the form
 \eqref{SS} to the canonical form $\Pi_{\lambda}=\omega_{\lambda}+\epsilon^2 \Theta+\mathcal{O}(\epsilon^3)$ one has to peform a Miura transformation producing (in general) infinitely many terms in the right hand side of \eqref{SS}. For this reason (in general) Strachan-Szablikowski truncated pencils correspond in our framework to non truncated pencils.
\end{rmk}

\section{Lifts of Poisson structures}

Given a differentiable manifold $M$, there is a natural way for lifting tensor fields and affine connections from $M$ to its tangent bundle $TM$, viewed as a manifold itself.
Such a lift is named \textit{complete lift} and has been extensively studied by Yano and Kobayashi \cite{YK, YKII, YKIII}.
 In this section we apply this construction to Poisson tensors defined on a suitable loop space. 

\subsection{Complete lift}
Let us recall the definition and some properties of complete lift, referring to original papers mentioned above for more details.

Given local coordinates $u^1,\dots,u^n$ on $M$, let $u^1,\dots,u^n,v^1,\dots,v^n$ be the induced bundle coordinates on $TM$ so that any tangent vector on $M$ has the form $v^i \frac{\partial}{\partial u^i}$. The complete lift of a function $f$, a one form $\alpha = \alpha_i du^i$, and a vector field $X = X^i \frac{\partial}{\partial u^i}$ is defined respectively by
\begin{equation}\label{liftbasic}
\hat f = v^j \frac{\partial f}{\partial u^j}, \qquad
\hat\alpha = v^j \frac{\partial \alpha_i}{\partial u^j} du^i + \alpha_i dv^i , \qquad
\hat X = X^i \frac{\partial}{\partial u^i} + v^j \frac{\partial X^i}{\partial u^j}\frac{\partial}{\partial v^i}.
\end{equation}
It follows readily from these local expressions that $\alpha(X)$ lifts to $\hat \alpha (\hat X)$, and a commutator $[X,Y]$ lifts to $[\hat X,\hat Y]$.

Lifted vector fields (resp. one-forms) span the tangent (resp. cotangent) space of $TM$ at any point which does not belong to the zero section $\{v=0\}$.
As a consequence, one can define the complete lift $\hat K$ of any given tensor field $K$ just by imposing that any contraction with a vector field $X$ or a one-form $\alpha$ on $M$ lifts to the contraction of $\hat K$ with $\hat X$ or $\hat \alpha$. 
Then one check that exterior derivative and Lie derivative are invariant with respect to the complete lift, meaning that $d\xi$ lifts to $d\hat \xi$ for any differential form $\xi$ and that a Lie derivative $L_XK$ lifts to $L_{\hat X}\hat K$.

It may be useful to have at hand explicit expressions for some special classes of tensors.
In particular, the complete lift of a bilinear form $g = g_{ij} du^i \otimes du^j$ turns out to be
\begin{equation}\label{liftg}
\hat g = v^k \frac{\partial g_{ij}}{\partial u^k} du^i \otimes du^j 
+ g_{ij} du^i \otimes dv^j 
+ g_{ij} dv^i \otimes du^j,
\end{equation}
and a trilinear form $T = T_{ijk} du^i \otimes du^j \otimes du^k$ lifts to
\begin{equation*}\label{liftT}
\hat T = v^h \frac{\partial T_{ijk}}{\partial u^h} du^i \otimes du^j \otimes du^k 
+ T_{ijk}  du^i \otimes du^j \otimes dv^k 
+ T_{ijk} du^i \otimes dv^j \otimes du^k 
+ T_{ijk} dv^i \otimes du^j \otimes du^k.
\end{equation*}
Moreover, an endomorphism of the tangent bundle $A=A^i_j \frac{\partial}{\partial u^i} \otimes du^j$ lifts to
\begin{equation}\label{liftA}
\hat A = A^i_j \frac{\partial}{\partial u^i} \otimes du^j 
+ v^k \frac{\partial A^i_j}{\partial u^k} \frac{\partial}{\partial v^i} \otimes du^j 
+ A^i_j \frac{\partial}{\partial v^i} \otimes dv^j,
\end{equation}
and the lift of a bilinear product on vector fields $\cdot = c^i_{jk} \frac{\partial}{\partial u^i} \otimes du^j \otimes du^k$ is
\begin{multline}\label{liftdot}
\hat \cdot = c^i_{jk} \frac{\partial}{\partial u^i} \otimes du^j \otimes du^k 
+ v^h\frac{\partial c^i_{jk}}{\partial u^h} \frac{\partial}{\partial v^i} \otimes du^j \otimes du^k \\
+ c^i_{jk} \frac{\partial}{\partial v^i} \otimes du^j \otimes du^k
+ c^i_{jk} \frac{\partial}{\partial v^i} \otimes dv^j \otimes du^k.
\end{multline}
Finally, any bivector $P = P^{ij} \frac{\partial}{\partial u^i} \otimes \frac{\partial}{\partial u^j}$ lifts to
\begin{equation}\label{liftbiv}
\hat P = P^{ij} \frac{\partial}{\partial u^i} \otimes \frac{\partial}{\partial v^j} 
+P^{ij} \frac{\partial}{\partial v^i} \otimes \frac{\partial}{\partial u^j} 
+ v^k \frac{\partial P^{ij}}{\partial u^k} \frac{\partial}{\partial v^i} \otimes \frac{\partial}{\partial v^j}.
\end{equation}

Let now $\nabla \frac{\partial}{\partial u^k} = \Gamma^i_{jk} \frac{\partial}{\partial u^i}\otimes du^j$ be an affine connection on $M$. 
Its complete lift $\hat \nabla$ is an affine connection on $TM$ defined by requiring that for all vector fields $X$ on $M$ the endomorphism $\nabla X$ lifts to $\hat \nabla \hat X$.
Using that $\frac{\partial}{\partial u^k}$ and $u^l\frac{\partial}{\partial u^k}$ lift to $\frac{\partial}{\partial u^k}$ and $u^l\frac{\partial}{\partial u^k}+v^l\frac{\partial}{\partial v^k}$ respectively,  one can check that
\begin{eqnarray}
\hat \nabla \frac{\partial}{\partial u^k} \label{liftGammau}
&=& \Gamma^i_{jk} \frac{\partial}{\partial u^i}\otimes du^j + v^h \frac{\partial\Gamma^i_{jk}}{\partial u^h}  \frac{\partial}{\partial v^i}\otimes du^j + \Gamma^i_{jk} \frac{\partial}{\partial v^i}\otimes dv^j, \\
\hat \nabla \frac{\partial}{\partial v^k} \label{liftGammav}
&=& \Gamma^i_{jk} \frac{\partial}{\partial v^i}\otimes du^j.
\end{eqnarray}
Readily from definition one deduces that for any tensor field $K$ on $M$ the complete lift of $\nabla K$ equals $\hat \nabla \hat K$.
In particular, any flat tensor $(\nabla K = 0)$ lifts to a flat tensor $(\hat \nabla \hat K=0)$.
Moreover it holds the following \cite[Proposition 7.1]{YK}:
\begin{prop}\label{liftTR}
The torsion and the curvature of $\hat \nabla$ are the complete lift of the torsion and the curvature of $\nabla$.
\end{prop} 

\begin{rmk}
Since the lift it is well defined for tensors and connections we can apply it to the geometric structures defining a Frobenius manifolds.
 As a result one obtain a \emph{lifted Frobenius structure}. We discuss this construction in more detail in the Appendix \ref{liftFrob}. 
\end{rmk}

\subsection{Lift of Poisson structures of hydrodynamic type}

The class of structures that can be lifted to the tangent bundle by means of complete lift includes symplectic forms and more generally Poisson tensors. 
The latter has been studied in some detail by Mitric and Vaisman \cite{MV}.
Since the Schouten bracket is defined in terms of Lie derivative, if follows that it is invariant by complete lift as well.
As a consequence, the complete lift of a bi-Hamiltonian structure $P_\lambda = P + \lambda Q$, where $\lambda \in \mathbf R$ and $P,Q$ are Poisson tensors on $M$ satisfying $[P,Q]=0$, is a bi-Hamiltonian structure $\hat P_\lambda = \hat P + \lambda \hat Q$.

Recall that, in local coordinates $u^i$ on $M$ and $x$ on $S^1$ the Poisson tensor $P$ at $\gamma=u(x)$ is represented by $\frac{\partial}{\partial u^i} \otimes P^{ij} \frac{\partial}{\partial u^j}$ where
\begin{equation}\label{componentofP}
P^{ij} = g^{ij} \partial_x + b^{ij}_k u^k_x,\qquad i,j=1,...,n.
\end{equation}
Here $g^{ij}$ is the inverse of the matrix $g_{ij}$ which represents $g$ locally, and $b^{ij}_k = -g^{ih} \Gamma^j_{hk}$, being $\Gamma^j_{hk}$ the Christoffel symbols of $g$. It is clear that $P$ can be lifted to $\mathcal L(TM)$ defining $\hat{P}$ as
$$
\hat{P}^{\alpha\beta}=\hat{g}^{\alpha\beta}\d_x+\hat{b}^{\alpha\beta}_{\gamma}u^{\gamma}_x,\qquad \alpha,\beta=1,...,2n, 
$$ 
where $\hat{g}$ is the lift of the contravariant metric, $\hat{b}^{\alpha\beta}_{\gamma}$ are the contravariant Christoffel symbols of the lifted Levi-Civita connection and we set $u^{n+i}=v^i$. Indeed one has only to check that $\hat \nabla$ is the Levi-Civita connection of the lifted metric $\hat g$.
But this follows by uniqueness of Levi-Civita connection together with the fact that $\hat \nabla \hat g = 0$ for $\nabla g=0$, and that $\hat \nabla$ is torsion free by Proposition \ref{liftTR} and by torsion-freeness of $\nabla$.  
Therefore $\hat g$ defines a Poisson structure of hydrodynamic type $\hat P$ on $\mathcal L(TM)$. 

\begin{rmk}
It is easy to check that the lift $\hat{P}$ is uniquely defined by the requirement (the analogous property in the finite dimensional case has been observed in \cite{MV})
\begin{equation}\label{main}
\{H_\xi,H_\eta\}_{\hat{P}}=\int_{S^1} \langle v,\{\xi,\eta\}_P\rangle\,dx
\end{equation}
where $H_{\xi}=\int_{S^1}\langle \xi, v\rangle\,dx$ and $\{\cdot,\cdot\}_P$ is the Poisson bracket on 1-forms \cite{GD,MM} defined by
 $g$ \cite{ALham}:
\beq\label{poissoncoordinatesf}
\{\xi, \eta\}_j=g^{kl}\left[\p_x^{s+1}(\eta)_l\f{\p (\xi)_j}{\p u^k_{(s)}}-\p_x^{s+1}(\xi)_l\f{\p (\eta)_j}{\p u^k_{(s)}}\right].
\eeq
\end{rmk}

\begin{prop}
In local coordinates $u^i,v^i$ on $TM$ one has
\begin{multline}\label{liftpoisson}
\hat P = \frac{\partial}{\partial v^i} \otimes (g^{ij} \partial_x + b^{ij}_ku^k_x) \frac{\partial}{\partial u^j} 
+ \frac{\partial}{\partial u^i} \otimes (g^{ij} \partial_x + b^{ij}_k u^k_x) \frac{\partial}{\partial v^j} \\
+ \frac{\partial}{\partial v^i} \otimes \left(v^h(b^{ij}_h+b^{ji}_h) \partial_x 
+ v^h \frac{\partial b^{ij}_k}{\partial u^h}u^k_x 
+ b^{ij}_k v^k_x \right) \frac{\partial}{\partial v^j}
\end{multline}
\end{prop}
\begin{proof}
Thanks to \eqref{componentofP} we have to determine the coefficients $g^{ij}$ and $b^{ij}_k$ for the lifted metric $\hat g$. 
To this end, let $W^j$ be the metric dual of the coordinate one-form $du^j$ on $M$.
This means that $W^j$ is the unique vector field on $M$ such that $g(W^j,\cdot) = du^j$, and clearly one has
\begin{equation}
W^j = g^{ij}\frac{\partial}{\partial u^i}.
\end{equation}
Moreover, well known properties of Christoffel symbols yield 
\begin{equation}\label{nablaW}
\nabla W^j = b^{ij}_k \frac{\partial}{\partial u^i} \otimes du^k.
\end{equation}
Therefore one can write
\begin{equation}
P = W^j \otimes \partial_x \frac{\partial}{\partial u^j} + \nabla_{\dot \gamma} W^j \otimes \frac{\partial}{\partial u^j},
\end{equation}
wehere $\dot \gamma = u^k_x \frac{\partial}{\partial u^k}$.

Let $U^j$ and $V^j$ be the metric dual of $du^j$ and $dv^j$ with respect to the lifted metric $\hat g$ on $TM$.
One can readily check by \eqref{liftg} that 
\begin{equation}\label{expressionUj}
U^j = g^{ij} \frac{\partial}{\partial v^i}.
\end{equation}
On the other hand, by \eqref{liftbasic} the lift of $du^j$ turns out to be $dv^j$.
Therefore $V^j = \hat W^j$, so that
\begin{equation}\label{expressionVj}
V^j = g^{ij}\frac{\partial}{\partial u^i} + v^k(b^{ij}_k + b^{ji}_k)\frac{\partial}{\partial v^i},
\end{equation}
where we used the identity
\begin{equation}\label{derivginv}
\frac{\partial g^{ij}}{\partial u^k} = b^{ij}_k + b^{ji}_k.
\end{equation}
In particular $\hat \nabla V^j = \hat \nabla \hat W^j $, whence by definition of lifted connection and equations \eqref{nablaW}, \eqref{liftA} it follows
\begin{equation}\label{hatnablaVj}
\hat \nabla V^j = b^{ij}_k \frac{\partial}{\partial u^i} \otimes du^k 
+ v^h \frac{\partial b^{ij}_k}{\partial u^h} \frac{\partial}{\partial v^i} \otimes du^k 
+ b^{ij}_k \frac{\partial}{\partial v^i} \otimes dv^k.
\end{equation}
On the other hand, by \eqref{liftGammav} one calculates
\begin{equation}
\hat \nabla U^j
= \frac{\partial g^{ij}}{\partial u^k} \frac{\partial}{\partial v^i} \otimes du^k
+ g^{ij} \Gamma^h_{ki} \frac{\partial}{\partial v^h}\otimes du^k,
\end{equation}
whence, thanks to the identity \eqref{derivginv}, one concludes
\begin{equation}\label{hatnablaUj}
\hat \nabla U^j
= b^{ij}_k \frac{\partial}{\partial v^i} \otimes du^k.
\end{equation}
The statement then follows by simple calculations from equations \eqref{expressionUj}, \eqref{expressionVj}, \eqref{hatnablaVj}, \eqref{hatnablaUj} and the identity
\begin{equation}
\hat P = 
U^j \otimes \partial_x \frac{\partial}{\partial u^j} + \hat \nabla_{\dot \gamma} U^j \otimes \frac{\partial}{\partial u^j} 
+ V^j \otimes \partial_x \frac{\partial}{\partial v^j} + \nabla_{\dot \gamma} V^j \otimes \frac{\partial}{\partial v^j},
\end{equation}
where $\dot \gamma = u^k_x \frac{\partial}{\partial u^k} + v^k_x \frac{\partial}{\partial v^k}$ for any loop $\gamma = (u(x),v(x))$ in $TM$.
\end{proof}

\subsection{Lift of bivectors in the loop space}
In matrix notation the lift \eqref{liftpoisson} takes the form
\begin{equation}\label{liftP}
\hat{P}
=
\begin{pmatrix}
0 & P^{ij} \cr
P^{ij} & \sum_{k,t}v^k_{(t)}\frac{\d P^{ij}}{\d u^k_{(t)}}
\end{pmatrix},
\end{equation}
whence it is clear that one can lift to $\mathcal L(TM)$ any given Poisson structure (non-necessarily of hydrodynamic type) on the loop space $\mathcal L(M)$. The proof of this fact is contained in the book \cite{Ku} in the framework of linearization of Hamiltonian objects a.k.a. formal or universal linearization  (see for instance \cite{K,KKVV}) or tangent covering (see for instance \cite{KV}).
 We provide here a different direct proof which rests just on the Schouten bracket formula given in \cite{DZ}.

\begin{thm}\label{thm::liftSchouten}
Suppose that 
\begin{equation*}
P^{ij}_{x,y}=P^{ij}_k(x-y,u, u_x, \ldots, u_{k+1})=\sum_{m=0}^{k+1} A^{ij}_m (u, u_x, \ldots, u_{k+1}) \delta^{(k+1-m)}(x-y).
\end{equation*}
and
\begin{equation*}
Q^{ij}_{x,y}=Q^{ij}_k(x-y,u, u_x, \ldots, u_{k+1})=\sum_{m=0}^{k+1} B^{ij}_m (u, u_x, \ldots, u_{k+1}) \delta^{(k+1-m)}(x-y).
\end{equation*}
have vanishing Schouten bracket 
\begin{eqnarray*}
 [P,Q]^{ijk}_{x,y,z} &=&
\frac{\partial P^{ij}_{x,y}}{\partial u^{l}_{(s)}(x)} \partial_x^s Q^{lk}_{x,z}
+ \frac{\partial Q^{ij}_{x,y}}{\partial u^{l}_{(s)}(x)} \partial_x^s P^{lk}_{x,z}
+\frac{\partial P^{ki}_{z,x}}{\partial u^{l}_{(s)}(z)} \partial_z^s Q^{lj}_{z,y}+\\
&&+ \frac{\partial Q^{ki}_{z,x}}{\partial u^{l}_{(s)}(z)} \partial_z^s P^{lj}_{z,y}
+\frac{\partial P^{jk}_{y,z}}{\partial u^{l}_{(s)}(y)} \partial_y^s Q^{li}_{y,x}
+ \frac{\partial Q^{jk}_{y,z}}{\partial u^{l}_{(s)}(y)} \partial_y^s P^{li}_{y,x}=0,
\end{eqnarray*}
then also the lifted structures
$$
\hat{P}=
\begin{pmatrix}
0 & P^{ij} \cr
P^{ij} & \sum_{k,t}v^k_{(t)}\frac{\d P^{ij}}{\d u^k_{(t)}}
\end{pmatrix},\qquad
\hat{Q}=
\begin{pmatrix}
0 & Q^{ij} \cr
Q^{ij} & \sum_{k,t}v^k_{(t)}\frac{\d Q^{ij}}{\d u^k_{(t)}}
\end{pmatrix}
$$
have vanishing Schouten bracket.
\end{thm}

\begin{proof}
Throughout in this proof $u^{n+i}$ will denote $v^i$ for all $i=1,\dots,n$.
Moreover we fix the convention that latin indices $i,j,k$ run from $1$ through $n$, and greek indices $\alpha,\beta,\gamma$ run from $1$ through $2n$.
By straightforward computation we obtain
\begin{itemize}
\item For $\alpha=i,\beta=j,\gamma=k$:
\begin{eqnarray*}
 [\hat{P},\hat{Q}]^{\alpha\beta\gamma}_{x,y,z} &=&
\frac{\partial \hat{P}^{ij}_{x,y}}{\partial u^{\lambda}_{(s)}(x)} \partial_x^s \hat{Q}^{\lambda k}_{x,z}
+ \frac{\partial \hat{Q}^{ij}_{x,y}}{\partial u^{\lambda}_{(s)}(x)} \partial_x^s \hat{P}^{\lambda k}_{x,z}
+\frac{\partial \hat{P}^{ki}_{z,x}}{\partial u^{\lambda}_{(s)}(z)} \partial_z^s \hat{Q}^{\lambda j}_{z,y}+\\
&&+ \frac{\partial \hat{Q}^{ki}_{z,x}}{\partial u^{\lambda}_{(s)}(z)} \partial_z^s \hat{P}^{\lambda j}_{z,y}
+\frac{\partial \hat{P}^{jk}_{y,z}}{\partial u^{\lambda}_{(s)}(y)} \partial_y^s \hat{Q}^{\lambda i}_{y,x}
+ \frac{\partial \hat{Q}^{jk}_{y,z}}{\partial u^{\lambda}_{(s)}(y)} \partial_y^s \hat{P}^{\lambda i}_{y,x}=0,
\end{eqnarray*}
since $\hat{P}^{ij}_{x,y}= \hat{Q}^{ij}_{x,y}= \hat{P}^{ki}_{z,x}= \hat{Q}^{ki}_{z,x}=\hat{P}^{jk}_{y,z}=\hat{Q}^{jk}_{y,z}=0$.
\item  For $\alpha=n+i,\beta=j,\gamma=k$:
\begin{eqnarray*}
 [\hat{P},\hat{Q}]^{\alpha\beta\gamma}_{x,y,z} &=&
\frac{\partial \hat{P}^{n+i,j}_{x,y}}{\partial u^{\lambda}_{(s)}(x)} \partial_x^s \hat{Q}^{\lambda k}_{x,z}
+ \frac{\partial \hat{Q}^{n+i,j}_{x,y}}{\partial u^{\lambda}_{(s)}(x)} \partial_x^s \hat{P}^{\lambda k}_{x,z}
+\frac{\partial \hat{P}^{k,n+i}_{z,x}}{\partial u^{\lambda}_{(s)}(z)} \partial_z^s \hat{Q}^{\lambda j}_{z,y}+\\
&&+ \frac{\partial \hat{Q}^{k,n+i}_{z,x}}{\partial u^{\lambda}_{(s)}(z)} \partial_z^s \hat{P}^{\lambda j}_{z,y}
+\frac{\partial \hat{P}^{jk}_{y,z}}{\partial u^{\lambda}_{(s)}(y)} \partial_y^s \hat{Q}^{\lambda,n+i}_{y,x}
+ \frac{\partial \hat{Q}^{jk}_{y,z}}{\partial u^{\lambda}_{(s)}(y)} \partial_y^s \hat{P}^{\lambda,n+i}_{y,x}=\\
&&\frac{\partial P^{ij}_{x,y}}{\partial u^{n+l}_{(s)}(x)} \partial_x^s Q^{lk}_{x,z}
+ \frac{\partial Q^{ij}_{x,y}}{\partial u^{n+l}_{(s)}(x)} \partial_x^s P^{lk}_{x,z}
+\frac{\partial P^{ki}_{z,x}}{\partial u^{n+l}_{(s)}(z)} \partial_z^s Q^{lj}_{z,y}+\\
&&+ \frac{\partial Q^{ki}_{z,x}}{\partial u^{n+l}_{(s)}(z)} \partial_z^s P^{lj}_{z,y}
+\frac{\partial \hat{P}^{jk}_{y,z}}{\partial u^{l}_{(s)}(y)} \partial_y^s Q^{li}_{y,x}
+ \frac{\partial \hat{Q}^{jk}_{y,z}}{\partial u^{l}_{(s)}(y)} \partial_y^s P^{li}_{y,x}
=0,
\end{eqnarray*}
since $\hat{P}^{jk}_{y,z}=\hat{Q}^{jk}_{y,z}=0$ and $ P^{ij}_{x,y},\, Q^{ij}_{x,y},\, P^{ki}_{x,y},\, Q^{ki}_{x,y}$
 do not depend on coordinates on the fibers. Similarly one can prove the vanishing of the Schouten bracket for
 $\alpha=i,\beta=n+j,\gamma=k$ and $\alpha=i,\beta=j,\gamma=n+k$.
\item  For $\alpha=n+i,\beta=n+j,\gamma=k$:
\begin{eqnarray*}
 [\hat{P},\hat{Q}]^{\alpha\beta\gamma}_{x,y,z} &=&
\frac{\partial \hat{P}^{n+i,n+j}_{x,y}}{\partial u^{\lambda}_{(s)}(x)} \partial_x^s \hat{Q}^{\lambda k}_{x,z}
+ \frac{\partial \hat{Q}^{n+i,n+j}_{x,y}}{\partial u^{\lambda}_{(s)}(x)} \partial_x^s \hat{P}^{\lambda k}_{x,z}
+\frac{\partial \hat{P}^{k,n+i}_{z,x}}{\partial u^{\lambda}_{(s)}(z)} \partial_z^s \hat{Q}^{\lambda,n+j}_{z,y}+\\
&&+ \frac{\partial \hat{Q}^{k,n+i}_{z,x}}{\partial u^{\lambda}_{(s)}(z)} \partial_z^s \hat{P}^{\lambda,n+j}_{z,y}
+\frac{\partial \hat{P}^{n+j,k}_{y,z}}{\partial u^{\lambda}_{(s)}(y)} \partial_y^s \hat{Q}^{\lambda,n+i}_{y,x}
+ \frac{\partial \hat{Q}^{n+j,k}_{y,z}}{\partial u^{\lambda}_{(s)}(y)} \partial_y^s \hat{P}^{\lambda,n+i}_{y,x}=\\
&&\frac{\partial \hat{P}^{n+i,n+j}_{x,y}}{\partial u^{n+l}_{(s)}(x)} \partial_x^s Q^{lk}_{x,z}
+ \frac{\partial \hat{Q}^{n+i,n+j}_{x,y}}{\partial u^{n+l}_{(s)}(x)} \partial_x^s P^{lk}_{x,z}
+\frac{\partial P^{ki}_{z,x}}{\partial u^{l}_{(s)}(z)} \partial_z^s Q^{lj}_{z,y}+\\
&&+ \frac{\partial Q^{ki}_{z,x}}{\partial u^{l}_{(s)}(z)} \partial_z^s P^{lj}_{z,y}
+\frac{\partial P^{jk}_{y,z}}{\partial u^{l}_{(s)}(y)} \partial_y^s Q^{li}_{y,x}
+ \frac{\partial Q^{jk}_{y,z}}{\partial u^{l}_{(s)}(y)} \partial_y^s P^{li}_{y,x}.
\end{eqnarray*}
Using the identities 
\beq\label{ident}
\frac{\partial \hat{P}^{n+i,n+j}_{x,y}}{\partial u^{n+l}_{(s)}(x)}=\frac{\partial P^{ij}_{x,y}}{\partial u^{l}_{(s)}(x)},\qquad 
\frac{\partial \hat{Q}^{n+i,n+j}_{x,y}}{\partial u^{n+l}_{(s)}(x)}=\frac{\partial Q^{ij}_{x,y}}{\partial u^{l}_{(s)}(x)},
\eeq
 we finally get
$$ [\hat{P},\hat{Q}]^{n+i,n+j,k}_{x,y,z}=[P,Q]^{ijk}_{x,y,z}=0.$$
 Similarly one can prove the vanishing of the Schouten bracket for
 $\alpha=i,\beta=n+j,\gamma=n+k$ and $\alpha=n+i,\beta=j,\gamma=n+k$.
\item For $\alpha=n+i,\beta=n+j,\gamma=n+k$:
\begin{eqnarray*}
 [\hat{P},\hat{Q}]^{\alpha\beta\gamma}_{x,y,z} &=&
\frac{\partial \hat{P}^{n+i,n+j}_{x,y}}{\partial u^{\lambda}_{(s)}(x)} \partial_x^s \hat{Q}^{\lambda,n+k}_{x,z}
+ \frac{\partial \hat{Q}^{n+i,n+j}_{x,y}}{\partial u^{\lambda}_{(s)}(x)} \partial_x^s \hat{P}^{\lambda,n+k}_{x,z}
+\frac{\partial \hat{P}^{n+k,n+i}_{z,x}}{\partial u^{\lambda}_{(s)}(z)} \partial_z^s \hat{Q}^{\lambda,n+j}_{z,y}+\\
&&+ \frac{\partial \hat{Q}^{n+k,n+i}_{z,x}}{\partial u^{\lambda}_{(s)}(z)} \partial_z^s \hat{P}^{\lambda,n+j}_{z,y}
+\frac{\partial \hat{P}^{n+j,n+k}_{y,z}}{\partial u^{\lambda}_{(s)}(y)} \partial_y^s \hat{Q}^{\lambda,n+i}_{y,x}
+ \frac{\partial \hat{Q}^{n+j,n+k}_{y,z}}{\partial u^{\lambda}_{(s)}(y)} \partial_y^s \hat{P}^{\lambda,n+i}_{y,x}=\\
&&\frac{\partial \hat{P}^{n+i,n+j}_{x,y}}{\partial u^{l}_{(s)}(x)} \partial_x^s Q^{lk}_{x,z}
+ \frac{\partial \hat{Q}^{n+i,n+j}_{x,y}}{\partial u^{l}_{(s)}(x)} \partial_x^s P^{lk}_{x,z}
+\frac{\partial \hat{P}^{n+k,n+i}_{z,x}}{\partial u^{l}_{(s)}(z)} \partial_z^s Q^{lj}_{z,y}+\\
&&+ \frac{\partial \hat{Q}^{n+k,n+i}_{z,x}}{\partial u^{l}_{(s)}(z)} \partial_z^s P^{lj}_{z,y}
+\frac{\partial \hat{P}^{n+j,n+k}_{y,z}}{\partial u^{l}_{(s)}(y)} \partial_y^s Q^{li}_{y,x}
+ \frac{\partial \hat{Q}^{n+j,n+k}_{y,z}}{\partial u^{l}_{(s)}(y)} \partial_y^s P^{li}_{y,x}+\\
&&\frac{\partial \hat{P}^{n+i,n+j}_{x,y}}{\partial u^{n+l}_{(s)}(x)} \partial_x^s \hat{Q}^{n+l,n+k}_{x,z}
+ \frac{\partial \hat{Q}^{n+i,n+j}_{x,y}}{\partial u^{n+l}_{(s)}(x)} \partial_x^s \hat{P}^{n+l,n+k}_{x,z}
+\frac{\partial \hat{P}^{n+k,n+i}_{z,x}}{\partial u^{n+l}_{(s)}(z)} \partial_z^s \hat{Q}^{n+l,n+j}_{z,y}+\\
&&+ \frac{\partial \hat{Q}^{n+k,n+i}_{z,x}}{\partial u^{n+l}_{(s)}(z)} \partial_z^s \hat{P}^{n+l,n+j}_{z,y}
+\frac{\partial \hat{P}^{n+j,n+k}_{y,z}}{\partial u^{n+l}_{(s)}(y)} \partial_y^s \hat{Q}^{n+l,n+i}_{y,x}
+ \frac{\partial \hat{Q}^{n+j,n+k}_{y,z}}{\partial u^{n+l}_{(s)}(y)} \partial_y^s \hat{P}^{n+l,n+i}_{y,x}
\end{eqnarray*}
Using the identities \eqref{ident} and the fact that the operator $\d_x$ and the operator $\sum_{k,t}u^{n+k}_{(t)}\frac{\d}{\d u^k_{(t)}(x)}$ commute, as it is immediate to check using the identity 
$$\d_x\frac{\d}{\d u^k_{(t)}(x)}=\frac{\d}{\d u^k_{(t)}(x)}\d_x-\frac{\d}{\d u^k_{(t-1)}(x)},$$
we obtain
\begin{eqnarray*}
&& [\hat{P},\hat{Q}]^{n+i,n+j,n+k}_{x,y,z} =\\
&&\sum_{k,t}\left(u^{n+k}_{(t)}(x)\frac{\d}{\d u^k_{(t)}(x)}+u^{n+k}_{(t)}(y)\frac{\d}{\d u^k_{(t)}(y)}+
 u^{n+k}_{(t)}(z)\frac{\d}{\d u^k_{(t)}(z)}\right)[P,Q]^{ijk}_{x,y,z}=0
\end{eqnarray*}
since $[P,Q]^{ijk}_{x,y,z}=0$  by hypothesis.
\end{itemize}
\end{proof}

\begin{rmk}
Notice that the lift of bivectors \eqref{liftP} is obtained from \eqref{liftbiv} just replacing $\sum_j v^j\f{\d}{\d u^j}$
 with $\sum_{j,k} v^j_{(k)}\f{\d}{\d u^j_{(k)}}$. 
The lift of general tensor fields can be defined in exactly the same way. For instance the lift of functionals, one forms  and vector fields can be defined as
\begin{equation*}
\hat F = \int v^j \frac{\delta F}{\delta u^j}\,dx, \qquad
\hat\alpha = \sum_{j,k}v^j_{(k)}\frac{\partial \alpha_i}{\partial u^j_{(k)}} \delta u^i + \alpha_i \delta v^i , \qquad
\hat X = X^i \frac{\partial}{\partial u^i} + \sum_{j,k}v^j_{(k)} \frac{\partial X^i}{\partial u^j_{(k)}}\frac{\partial}{\partial v^i}.
\end{equation*}
As in the finite dimensional case the lift  $\hat K$ of higher order tensor fields $K$ can be defined requiring that any contraction with a vector field $X$ or a one-form $\alpha$ on the loop space lifts to the contraction of $\hat K$ with $\hat X$ or $\hat \alpha$. As a consequence
 of this general rule the lift of a Hamiltonian vector field coincides with the Hamiltonian vector field obtained lifting the Poisson bivector and the Hamiltonian functional:  $\widehat{P\delta H}=\hat{P}\delta\hat{H}$. In the Appendix C we check this fact. 
 Finally we point out that the linearization of Hamiltonian objects mentioned above is nothing but the
  Yano-Kobayashi complete lift  in the infinite-dimensional setting.
\end{rmk}

\subsection{Lift of deformations}
We have seen in the introduction that deformations of $n$-component semisimple Poisson pencils of hydrodynamic type depend on $n$ arbitrary functions of a single variable. Applying the previous construction to this case we get a $n$-parameter family
 of deformations of the lifted Poisson pencil of hydrodynamic type. Due to obvious identity
$${\rm det}\hat\pi^{ij}=\pm\left({\rm det}\pi^{ij}\right)^2$$
any invariant coefficient comes with double multiplicity. This example suggests that deformations of non semisimple structures corresponding to those invariant parameters are unobstructed.

\subsubsection{Example}
In the scalar case all second order deformations are given by \cite{L}
\begin{equation}\label{scalar_1}
\Pi_{\lambda}=2u\d_x+u_x -\lambda \d_x+ \epsilon^2 (2 s \partial_x^3 + 3 s_x \partial_x^2 + s_{xx} \partial_x) +\mathcal{O}(\epsilon^3),
\end{equation}
where $c$ is a constant and $s(u)$ is an arbitrary function of $u$. Applying the lift we obtain a one-parameter family of deformations of a 2-component Poisson pencil of hydrodynamic type.

Here we want to show this lift is equivalent, up to Miura transformations, to the case N3 (that is, N6 with $\kappa=1$) with $F_1(u^1)=\eta^{22}=0$. Let us consider second order deformations of $N3$ obtained in Theorem \ref{thm_quasi}, and set $\eta^{22}=0$ (otherwise
 $g_1$ would not be the lift of the scalar constant metric $\eta=1$), $\eta^{12}=1$, $F_1(u^1)=0$ and $F_2(u^1)=-\frac{f(u^1)}{u^1}$. 

The Miura transformation
$$
u^i \to \exp(-\epsilon Y) u^i, \quad i=1,2,
$$
 generated by the vector field $Y$ of components
$$
Y^1=\frac{f'}{3} u^1_{xx} + \frac{f''}{3} (u^1_x)^2,\quad
Y^2=  -\frac{f''}{3}   u^1_x u^2_x  -\frac{f'}{3}  u^2_{xx},
$$
 reduces the pencil to the form 
\begin{equation*}\label{pencil_lift}
\hat\Pi_{\lambda}=
\begin{pmatrix}
0 & \Pi_{\lambda}\\
\Pi_{\lambda} & \sum_{t} v_{(t)} \frac{\partial  \Pi_{\lambda}}{\partial u_{(t)}}
\end{pmatrix},
\end{equation*}
where $\Pi_{\lambda}$ coincides with \eqref{scalar_1}  setting $u^1=u$ and $f(u^1)=s(u)$.


\appendix

\section{Appendix: Computations of deformations}\label{appA}
In this appendix we give a sketch of the proof of Theorem \ref{thm_quasi}, providing the computations of deformations in detail. First of all we observe that the pencil $\Pi^{ij}_{\lambda}$ can be always reduced to the form
\beq\label{pencil_appendix}
\Pi_{\lambda}=\omega_{\lambda}+\epsilon Q_1+\epsilon^2 Q_2+\epsilon^3Q_3+...\\
\eeq
by a suitable Miura transformation. The proof is due to Getzler and it is based on the study of Poisson-Lichnerowicz cohomology groups
 \cite{G} (an alternative proof can be found in \cite{DMS,DZ,LZ2}) :
\begin{equation*}\label{eq33bis.eq}
H^j(\mathcal{L}(\mathbb{R}^n), \omega):=\frac{\ker\{d_{\omega}: \Lambda^j_{\text{loc}}\rightarrow \Lambda^{j+1}_{\text{loc}}\}}
{\mathrm{im}\{d_{\omega}: \Lambda^{j-1}_{\text{loc}} \rightarrow \Lambda^j_{\text{loc}}\}}
\end{equation*}
for Poisson bivector of hydrodynamic type $\omega$. The differential $d_{\omega}$ is defined as
$$d_{\omega}:=[\omega,\, \cdot \ ]$$
where the square bracket is the Schouten bracket. Getzler also proved the triviality of cohomology for any positive integer $j$ (in particular the triviality
 of deformations is related to the vanishing of the second cohomology group).
 
\subsection{First order deformations} The pencil \eqref{pencil_appendix} is a deformation of $\omega_{\lambda}$ if it satisfies the Jacobi identity for every $\lambda$, that is
$$
[Q,Q]=[\omega_1,Q]=0.
$$
where $Q=\omega_2+\epsilon Q_1+\epsilon^2 Q_2+\epsilon^3Q_3+...$. This implies in particular
$$[\omega_2, Q_1]=[\omega_1,Q_1]=0.$$
In other words $Q_1$ is a cocycle for both the differentials $d_{\omega_1}$ and $d_{\omega2}$. Using the triviality of
 $H^1(\mathcal{L}(\mathbb{R}^n), \omega)$ and $H^2(\mathcal{L}(\mathbb{R}^n), \omega)$ we obtain $Q_1=d_{\omega_2}X=
{\rm Lie}_X\omega_2$ for a suitable vector field of degree 1 
$$
X^i=X^i_1(u^1,u^2)u^1_x + X^i_2(u^1,u^2)u^2_x, \quad i=1,2,
$$
satisfying
$$
d_{\omega_1} d_{\omega_2} X=0.
$$
It is not difficult to prove that  among the solutions of the above equation those corresponding to trivial deformations have the form
 $X=\omega_1\delta H + \omega_2\delta K$, where the hamiltonian denisties are differential polynonials of degree 0, namely $H=\int h(u^1,u^2)\,dx$ and $K=\int k(u^1,u^2)\,dx$. It turns out that in our case all first order defomations are trivial. All details below, case by case.

\subsubsection{T3. First order deformations}
Let us point out that in this case the vanishing of the coefficient $\eta^{22}$ implies that the affinor $L^i_j$ assumes diagonal form, while for $\eta^{22}\neq0$ it corresponds to one $2\times 2$ Jordan block case (as well as all other cases we are dealing with).
Recall that we are assuming $\eta^{12}\neq0$. The vector field $X$ solution of $d_{\omega_1} d_{\omega_2}X=0$ is given in components by
\begin{gather*}
X^1_1=X^1_1, \quad X^1_2=X^1_2, \quad
X^2_1= \frac{\eta^{22}}{\eta^{12}} \partial_1 (X^1_1 u^1)+\int \left(\partial_1 X^1_1 -\frac{\eta^{22} u^1 }{\eta^{12}}\partial_1^2 X^2_1\right) \ du^2+F,\\
X^2_2=X^1_1+\frac{\eta^{22}}{\eta^{12}} \left( X^1_2+ u^1(\partial_2 X^1_1-\partial_1 X^1_2 ) \right),
\end{gather*}
where $F=F(u^1)$.
The components $Y^i$ of the vector field $Y=\omega_1\delta H +\omega_2 \delta K$ are given by $Y^i=Y^i_1 u^1_x+Y^i_2 u^2_x$, where
$$
\begin{array}{ll}
Y^1_1 = \eta^{12} \partial_1\partial_2 H -u^1\partial_1\partial_2 K, &
Y^1_2 = \eta^{12} \partial_2^2 H -u^1\partial_2^2 K,\\
Y^2_1 = \partial_1 (\eta^{12} \partial_1 H +\eta^{22} \partial_2 H- u^1\partial_1 K), &
Y^2_2 = \partial_2 (\eta^{12} \partial_1 H +\eta^{22} \partial_2 H- u^1\partial_1 K),
\end{array}
$$
Choosing $H$ and $K$ such that $X^1_i=Y^1_i$ for $i=1,2$, one can easily see that
$$
X^2_1=Y^2_1+F, \quad X^2_2=Y^2_2.
$$
Finally, the function $F$ can be removed using the vector field $Y$ such that $H=0$ and $K$ such that $-\partial_1(u^1\partial_1 K)=F$. Thus, first order deformations are trivial.

\subsubsection{N5.  First order deformations}
Here $\eta^{12}\neq0$. Solving $d_{\omega_1} d_{\omega_2}X=0$ for $\deg(X)=1$ we get
\begin{gather*}
X^1_2=  \partial_1 F, \quad  X^2_1= \partial_2 F,\\
X^2_1= \int \left( \partial_1 X^2_2 +\frac{\eta^{22} \partial_2 F + \eta^{12} \partial_1 F - \eta^{12} X^2_2}{2 \eta^{12} (u^1+u^2)- \eta^{22} u^1} \right)  \ du^2 + G,
\quad X^2_2=X^2_2,
\end{gather*}
where $F=F(u^1,u^2)$ and $G=G(u^1)$.

The components $Y^i$ of the vector field $Y=\omega_1\delta H + \omega_2 \delta K$ are given by
\begin{eqnarray*}
Y^1_1&=& \partial_1 (\eta^{12} \partial_2 H+ u^1 \partial_2 K),\\
Y^1_2&=& \partial_2 (\eta^{12} \partial_2 H+ u^1 \partial_2 K),\\
Y^2_1&=& \eta^{12} \partial_1^2 H + \eta^{22} \partial_1 \partial_2 H + u^1 \partial_1^2 K + 2 (u^1+u^2) \partial_1 \partial_2 K + \partial_2 K,\\
Y^2_2&=& \eta^{12} \partial_1 \partial_2 H + \eta^{22} \partial_2^2 H + u^1 \partial_1 \partial_2 K + 2 (u^1+u^2) \partial_2^2 K +  \partial_2 K.
\end{eqnarray*}
Choosing $H$ and $K$ such that $F=\eta^{12} \partial_2 H+ u^1 \partial_2 K, \ X^2_2=Y^2_2$,
we obtain
$$
X^1_1=X^1_2=X^2_2=0, \quad X^2_1=G.
$$
Taking $H=0$ and $K$ such that $\partial_2 K=0$ and $ u^1 \partial_1^2 K=G$, we can also remove $G$. Thus, deformations of degree 1 are trivial.

\subsubsection{N3, N4 and N6.  First order deformations}
This case is more involved. Let us assume $\kappa\neq-1$, otherwise the metric $g_2$ would be degenerate. Here $\eta^{12}\neq0$.


\medskip
\noindent
Imposing $d_{\omega_1} d_{\omega_2}X=0$ for $\deg(X)=1$ we get
\begin{gather*}
X^1_1=\partial_1 G+R, \quad X^1_2=\partial_2 G, \quad X^2_1= \partial_1 F, \quad X^2_2=\partial_2 F,\\
R=\theta^{\frac{\kappa}{2}}
\int  \kappa \left( \eta^{22} \partial_2 G +\eta^{12} \partial_1 G -\eta^{12} \partial_2 F
 \right) \theta^{-1-\frac{\kappa}{2}} \ du^2
+\theta^{\frac{\kappa}{2}} S,
\end{gather*}
where $F=F(u^1,u^2), G=G(u^1,u^2), S=S(u^1)$ and $\theta=2 \eta^{12} u^2 - (1+\kappa) \eta^{22} u^1$.
The components $Y^i$ of the vector field $Y=\omega_1\delta H + \omega_2 \delta K$ are given by
\begin{eqnarray*}
Y^1_1&=&
\partial_1 (\eta^{12} \partial_2 H+(1+\kappa) u^1 \partial_2 K) - \kappa \partial_2 K,\\
Y^1_2&=&
\partial_2 (\eta^{12} \partial_2 H+(1+\kappa) u^1 \partial_2 K),\\
Y^2_1
&=& \partial_1 (\eta^{12} \partial_1 H + \eta^{22} \partial_2 H + 2 u^2 \partial_2 K + (1+\kappa) u^1 \partial_1 K- K),\\
Y^2_2
&=& \partial_2 (\eta^{12} \partial_1 H + \eta^{22}  \partial_2 H + 2 u^2 \partial_2 K + (1+\kappa) u^1 \partial_1 K- K).
\end{eqnarray*}
Choosing $H$ and $K$ such that
$$
\eta^{12} \partial_2 H+(1+\kappa) u^1 \partial_2 K=F,
$$
$$
\eta^{12} \partial_1 H + \eta^{22}  \partial_2 H + 2 u^2 \partial_2 K + (1+\kappa) u^1 \partial_1 K- K=G,
$$
we get
$$
X^1_1=\theta^{\frac{\kappa}{2}} S, \quad X^1_2=X^2_1=X^2_2=0.
$$
Finally, taking a suitable choose of $H$ and $K$, we can also remove $S$. In particular, we have
\begin{itemize}
\item for $\kappa\neq 0,-2$
$$
H=\frac{(1+\kappa)u^1 \theta^{1+\frac{\kappa}{2}} S }{(\eta^{12})^2 \kappa (\kappa+2)},
\quad
K=-\frac{ \theta^{1+\frac{\kappa}{2}} S} {\eta^{12} \kappa (\kappa+2)},
$$
\item for $\kappa=0$
$$
H=\frac{(2 \eta^{12} u^2 -\eta^{22} u^1)(\log(2 \eta^{12} u^2 -\eta^{22} u^1) -1) u^1 S }{4 (\eta^{12})^2},
$$
$$
K=  \frac{u^2 \int S \ du^1}{ u^1}    -  \frac{(2 \eta^{12} u^2 -\eta^{22} u^1)(\log(2 \eta^{12} u^2 -\eta^{22} u^1) -1) S   }{4 \eta^{12}}  -\int \!\!\! \int \frac{\eta^{22} \partial_1 (u^1 S)}{2 \eta^{12} u^1} \ du^1 \ du^1 
$$
\item for  $\kappa=-2$
$$
H=\frac{\log (2 \eta^{12}u^2 +\eta^{22} u^1) u^1 S  }{4 (\eta^{12})^2},
\quad
K=\frac{\log (2 \eta^{12}u^2 +\eta^{22} u^1) S  }{4 \eta^{12}}
+
\frac{\int S \ du^1}{2 \eta^{12} u^1}.
$$
\end{itemize}
Thus, first-order deformations are trivial.

\subsection{Second order deformations}\label{app_2}
We have seen that in all cases $Q_1$ can be eliminated by a Miura transformation.
 For this reason, without loss of generality, we can assume the pencil has the 
 form
$$
\Pi_{\lambda}=\omega_{\lambda}+\epsilon^2 Q_2+\epsilon^3Q_3+...\\
$$
Using the same arguments applied to first order deformations we can easily prove that 
\begin{itemize}
\item general second order deformations  can be always written as $Q_2=d_{\omega_2}X$ for a suitable vector field of degree 2 
$$
X^i=X^i_1(u^1,u^2)u^1_{xx} + X^i_2(u^1,u^2) (u^1_x)^2+X^i_3(u^1,u^2) u^1_x u^2_x + X^i_4(u^1,u^2) (u^2_x)^2+X^i_5(u^1,u^2) u^2_{xx},  
$$
satisfying
$$
d_{\omega_1} d_{\omega_2} X=0.
$$
\item trivial second order deformations are those corresponding to vector fields of the form $\omega_1\delta H + \omega_2\delta K$, where the hamiltonian
 functionals $H$ and $K$ have hamiltonian densities of degree 1, namely
$$
H=\int\left[h_1(u^1,u^2) u^1_x+h_2(u^1,u^2) u^2_x\right] dx, \quad K=\int \left[k_1(u^1,u^2) u^1_x+k_2(u^1,u^2) u^2_x\right]dx.
$$
\end{itemize}

Before to go into the details of the computations, let us observe that 
$$
\delta H =
\begin{pmatrix}
\dfrac{\delta H}{\delta u^1}\\[10pt]
\dfrac{\delta H}{\delta u^2}
\end{pmatrix}
=
\begin{pmatrix}
\dfrac{\partial H}{\partial u^1}-\dfrac{d}{dx} \dfrac{\partial H}{\partial u^1_x}\\[10pt]
\dfrac{\partial H}{\partial u^2}-\dfrac{d}{dx} \dfrac{\partial H}{\partial u^2_x}
\end{pmatrix}
=\begin{pmatrix}
R(u^1,u^2) u^2_x\\
-R(u^1,u^2) u^1_x
\end{pmatrix},
$$
for $R(u^1,u^2)=\partial_1 H_2(u^1,u^2)-\partial_2 H_1(u^1,u^2)$ and similarly
$$
\delta K =
\begin{pmatrix}
\dfrac{\delta K}{\delta u^1}\\[10pt]
\dfrac{\delta K}{\delta u^2}
\end{pmatrix}
=
\begin{pmatrix}
\dfrac{\partial K}{\partial u^1}-\dfrac{d}{dx} \dfrac{\partial K}{\partial u^1_x}\\[10pt]
\dfrac{\partial K}{\partial u^2}-\dfrac{d}{dx} \dfrac{\partial K}{\partial u^2_x}
\end{pmatrix}=
\begin{pmatrix}
S(u^1,u^2) u^2_x\\
-S(u^1,u^2) u^1_x
\end{pmatrix},
$$
for $S(u^1,u^2)=\partial_1 K_2(u^1,u^2)-\partial_2 K_1(u^1,u^2)$. 
\newline
\newline
 We now proceed as follows:
\begin{enumerate}
\item We solve the equation $d_{\omega_1} d_{\omega_2} X=0$, which leads to a solution depending on two functions of two variables and at most four functions of one variable.
\item Up to Miura-type transformations, that is, using the freedom given by the functions $R$ and $S$, we can eliminate the two functions of two variables.
\item In the cases T3, N3, N5 and N6 with $\kappa \neq -1, -2$, we still use a Miura-type transformation to reduce the deformation to a more suitable form (see step 4). 
\item The last step is quite straightforward. We firstly take a generic Hamiltonian vector field of the form $X=\omega_1\delta H-\omega_2\delta K$ with
$$
H=\int \sum_{i,j} \left( h_{ij} u^i_x \log u^j_x  \right) \ dx,\qquad K=\int \sum_{i,j} \left( k_{ij} u^i_x \log u^j_x  \right) \ dx,
$$
where the coefficients $h_{ij}$ and $k_{ij}$ are arbitrary functions of $(u^1,u^2)$. Then, comparing $X$ with the vector field obtained above (step 3), we get the values of $h_{ij}$ and $k_{ij}$ which correspond to the final expression written in Theorem \ref{thm_quasi}. 
\end{enumerate}

Let us discuss in detail each case. In what follows, all the functions $X^i_j, R, S$, $i=1,2$, $j=1,\ldots,5$, will depend on $(u^1,u^2)$, unless stated  otherwise.

\subsubsection{T3. Second order deformations}
Let us assume $\eta^{22} \neq 0$. The solution of  $d_{\omega_1} d_{\omega_2}X=0$ for $\deg(X)=2$ is given by
\begin{eqnarray*}
X^1_1&=&X^1_1,\\
X^1_2&=&X^1_2,\\
X^1_3&=& \frac{2}{3} \partial_2 X^1_1-\frac{1}{3} \partial_2 X^2_5,\\
X^1_4&=& 0, \\
X^1_5&=& 0,\\
X^2_1&=&\frac{\eta^{22} u^1}{\eta^{12}}\left(X^1_2-\frac{4}{3} \partial_1X^1_1 \right)-\frac{\eta^{22}}{3\eta^{12}}\left( \partial_1 (u^1 X^2_5)+2 X^2_5\right) +F_1,\\
X^2_2&=& \partial_1 X^2_1,\\
X^2_3&=& \partial_2 X^2_1+\partial_1 X^2_5,\\
X^2_4&=&\partial_2  X^2_5,\\
X^2_5&= &F_2 e^{\frac{-\eta^{12} u^2}{\eta^{22} u^1}}-X^1_1.
\end{eqnarray*}
where $F_1,F_2$ depend on $u^1$. The components $Y^i$ of the vector field $Y=\omega_1\delta H + \omega_2\delta K$ are
\begin{eqnarray*}
Y^1_1&=& -\eta^{12} R + u^1 S,\\
Y^1_2&=&-\eta^{12} \partial_1 R +u^1 \partial_1 S,\\
Y^1_3&=&-\eta^{12} \partial_2 R +u^1 \partial_2 S,\\
Y^1_4&=&0,\\
Y^1_5&=&0,\\
Y^2_1&=& -\eta^{22} R, \\
Y^2_2&=& -\eta^{22} \partial_1 R,\\
Y^2_3&=&\eta^{12} \partial_1 R -\eta^{22} \partial_2 R -u^1 \partial_1 S - S,\\
Y^2_4&=&\eta^{12} \partial_2 R -u^1 \partial_2 S,\\
Y^2_5&=& \eta^{12} R - u^1 S.
\end{eqnarray*}
Choosing $R$ and $S$ such that $X^1_i=Y^1_i$ for $i=1,2$, we finally obtain
\begin{eqnarray*}
X^1_1&=&0\\
X^1_2&=&0\\
X^1_3&=&-\frac{1}{3} \partial_2 X^2_5,\\
X^1_4&=& 0, \\
X^1_5&=& 0,\\
X^2_1&=&-\frac{\eta^{22}}{3\eta^{12}}\left( \partial_1 (u^1 X^2_5)+2 X^2_5\right)+F_1,\\
X^2_2&=& \partial_1 X^2_1,\\
X^2_3&=& \partial_2 X^2_1+\partial_1 X^2_5,\\
X^2_4&=&\partial_2  X^2_5,\\
X^2_5&= &F_2 e^{\frac{-\eta^{12} u^2}{\eta^{22} u^1}}.
\end{eqnarray*}
Thus, these coefficients depend on two functions $F_1$, $F_2$ in the variable $u^1$.

In the case $\eta^{22}=0$, the computation is easier.
The condition $d_{\omega_1} d_{\omega_2}X=0$ implies
\begin{eqnarray*}
X^1_1&=&X^1_1\\
X^1_2&=&X^1_2\\
X^1_3&=& \partial_2 X^1_1,\\
X^1_4&=& 0, \\
X^1_5&=& 0,\\
X^2_1&=& F,\\
X^2_2&=&  \partial_1 F,\\
X^2_3&=& -\partial_1 X^1_1,\\
X^2_4&=&-\partial_2  X^1_1,\\
X^2_5&= &-X^1_1.
\end{eqnarray*}
where $F$ depends on $u^1$. Also in this case the freedom in $R$ and $S$ allows us to reduce $X^1_1$ and $X^1_2$ to zero, obtaining

$$
X^1=0,
\quad
X^2=F u^1_{xx}+\partial_1 F (u^1_x)^2 = (F u^1_x)_x
$$
The second component of the vector field can be written as
$$
X_2=\d_x^2 \int F \ du^1,
$$
and setting $ f= Fu^1$ yields
$$
Q_2=
\begin{pmatrix}
0 & 0\\
0& f_{xx} \delta' + 3 f_x \delta'' +2 f \delta'''
\end{pmatrix}.
$$

Finally, in order to get the form we need to compute $h_{ij}$ (step 3), 
we perform the  canonical Miura transformation generated by the local Hamiltonian 
$$
H=-\int_{S^1}\left(\frac{\eta^{22} (u^1)^2F_2'}{3 (\eta^{12})^2} +\frac{u^2 F_2}{3 \eta^{12}}\right) e^{-\frac{\eta^{12} u^2}{\eta^{22} u^1}} \  u^1_x\,dx.
$$

\begin{rmk}
Let us point out that this solution can be obtained from the general case in the limit $\eta^{22} \rightarrow 0$. 
\end{rmk}


\subsubsection{N5. Second order deformations}
The condition $d_{\omega_1} d_{\omega_2}X=0$ for $\deg(X)=2$ implies
\begin{eqnarray*}
X^1_1&=& X^1_1,\\
X^1_2&=& \partial_1 X^1_1,\\
X^1_3&=& \partial_2 X^1_1,\\
X^1_4&=& 0,\\
X^1_5&=& 0,\\
X^2_1&=& X^2_1,\\
X^2_2&=& \partial_1 X^2_1 + \frac{2}{3} \theta^{1/2} \partial_1 F_2+ \frac{5 \eta^{12} - 2 \eta^{22}}{3} \theta^{3/2} F_2+ \theta (\eta^{22}  X^1_1- \eta^{12}  X^2_1+ F_1),\\
X^2_3&=& \partial_2 X^2_1 -\partial_1 X^1_1 + \theta^{1/2} \partial_1 F_2-\frac{4\eta^{12} - 3 \eta^{22}}{6}\theta^{3/2} F_2,\\
X^2_4&=& -\eta^{12} \theta^{3/2} F_2-\partial_2 X^1_1,\\
X^2_5&=& \theta^{1/2} F_2 - X^1_1,
\end{eqnarray*}
where $F_i$, for $i=1,2$, are functions depending on $u^1$ and $\theta=(2 \eta^{12} (u^1+u^2) -\eta^{22} u^1)^{-1}$.
The components $Y^i$ of the vector field $Y=\omega_1\delta H +\omega_2\delta K$ are
\begin{eqnarray*}
Y^1_1&=& -(\eta^{12} R + u^1 S),\\
Y^1_2&=& -\partial_1(\eta^{12} R + u^1 S),\\
Y^1_3&=& -\partial_2(\eta^{12} R + u^1 S),\\
Y^1_4&=& 0,\\
Y^1_5&=& 0,\\
Y^2_1&=& - (\eta^{22} R + 2(u^1+u^2)S),\\
Y^2_2&=&- (\eta^{22} \partial_1 R + 2 (u^1+u^2)\partial_1 S +S),\\
Y^2_3&=&
\partial_1 (\eta^{12} R + u^1 S) - \partial_2(\eta^{22} R + 2(u^1+u^2)S),\\
Y^2_4&=& \partial_2 (\eta^{12} R + u^1 S),\\
Y^2_5&=& \eta^{12} R + u^1 S.
\end{eqnarray*}
Choosing $R, S$ such that $X^i_1=Y^i_1$ for $i=1,2$, we can reduce $X^1$ to zero and the coefficients of $X^2$ respectively to
\begin{eqnarray*}
X^2_1&=& 0,\\
X^2_2&=& \frac{2}{3} \partial_1 (\theta^{1/2} F_2) -\frac{7}{3}\partial_2(\theta^{1/2} F_2)-\eta^{22} \theta^{3/2} F_2+ \theta F_1,\\
X^2_3&=& \partial_1 (\theta^{1/2} F_2)-\frac{1}{3}\partial_2(\theta^{1/2} F_2),\\
X^2_4&=& \partial_2 (\theta^{1/2} F_2),\\
X^2_5&=& \theta^{1/2} F_2.
\end{eqnarray*}
Thus, the deformations of degree 2 depend on two functions of $u^1$.

To reduce the deformation in the form written in Theorem \ref{thm_quasi} (step 3) we perform the canonical Miura transformation generated by
$$
H= \int_{S^1}\frac{u^1}{(\eta^{12})^2}\left(\frac{(3 \eta^{22} -8\eta^{12})\theta^{1/2}F_2}{6} +\theta^{-1/2}F_2' +\frac{ \log (\theta^{-1}) F_1 }{2}\right)u^1_x\,dx.
$$


\subsubsection{N3, N4 and N6. Second order deformations}

The vector fields $Y=P\delta H + Q \delta K$ are given by
\begin{eqnarray*}
Y^1_1&=& -(\eta^{12} R + (1+\kappa) u^1 S),\\
Y^1_2&=& -\partial_1 (\eta^{12}  R + (1+\kappa) u^1 S) +\kappa S,\\
Y^1_3&=& -\partial_2 (\eta^{12} R + (1+\kappa) u^1 S),\\
Y^1_4&=& 0,\\
Y^1_5&=& 0,\\
Y^2_1&=& -(\eta^{22} R + 2 u^2 S),\\
Y^2_2&=& -\partial_1 (\eta^{22}  R + 2 u^2 S),\\
Y^2_3
&=& \partial_1 (\eta^{12}  R + (1+\kappa) u^1 S)-\partial_2 (\eta^{22}  R + 2 u^2 S),\\
Y^2_4&=& \partial_2 (\eta^{12} R + (1+\kappa) u^1 S),\\
Y^2_5&=& \eta^{12} R + (1+\kappa) u^1 S.
\end{eqnarray*}

In studynig the solutions of the equation $d_{\omega_1} d_{\omega_2}X=0$ we have to distinguish 3 cases: $\kappa=0,\,\kappa=-2,\,\kappa\ne 0,2$. This is due to the fact that conditions coming from this equation include the following:
$$
\kappa (\kappa+2) X^1_5(u^1,u^2)=0.
$$

\medskip
\noindent
\emph{Case 1:} $\kappa=0$.
The condition $d_{\omega_1} d_{\omega_2}X=0$ for $\deg(X)=2$ leads to
\begin{eqnarray*}
X^1_1&=& X^1_1,\\
X^1_2&=& \partial_1 X^1_1 + \theta F_1, \\
X^1_3&=& \theta \partial_1 F_2- \eta^{22} \theta^{2} F_2 +\partial_2 X^1_1,\\
X^1_4&=& 2 \eta^{12} \theta^2 F_2, \\
X^1_5&=& \theta F_2, \\
X^2_1&=& X^2_1,\\
X^2_2&=&\partial_1 X^2_1 + \theta F_3, \\
X^2_3&=&  - \frac{\partial_1^2 F_2}{\eta^{12}}+\theta^{\frac{1}{2}} \partial_1 F_4 - \frac{\eta^{22}}{2} \theta^{\frac{3}{2}} F_4- \partial_1 X^1_1 + \partial_2 X^2_1,\\
X^2_4&=&  \eta^{12} \theta^{\frac{3}{2}} F_4 -\partial_2 X^1_1,\\
X^2_5&=& - \frac{\partial_1 F_2}{\eta^{12} } + \theta^{\frac{1}{2}} F_4- X^1_1,
\end{eqnarray*}
where $F_i$ for $i=1,\ldots,4$ are arbitrary functions depending on $u^1$, and $\theta=(\eta^{22}  u^1-2 \eta^{12}  u^2)^{-1}$.
Choosing $R$ and $S$ such that $X^i_1=Y^i_1$ for $i=1,2$, we can reduce both $X^i_1$, $i=1,2$, to zero, obtaining
\begin{eqnarray*}
X^1_1&=& 0,\\
X^1_2&=& \theta F_1, \\
X^1_3&=& \partial_1 (\theta F_2) \\
X^1_4&=& \partial_2 (\theta F_2), \\
X^1_5&=& \theta F_2, \\
X^2_1&=& 0,\\
X^2_2&=& \theta F_3, \\
X^2_3&=&  \partial_1\left(\theta^{\frac{1}{2}} F_4-\frac{\partial_1 F_2}{\eta^{12} }\right),\\
X^2_4&=&  \partial_2\left(\theta^{\frac{1}{2}} F_4-\frac{\partial_1 F_2}{\eta^{12} }\right),\\
X^2_5&=&  \theta^{\frac{1}{2}} F_4- \frac{\partial_1 F_2}{\eta^{12} }.
\end{eqnarray*}
In this case, the deformations of degree 2 depend on four functions on $u^1$.

\medskip
\noindent
\emph{Case 2:} $\kappa=-2$.
The condition $d_{\omega_1}d_{\omega_2}X=0$ for $\deg(X)=2$ implies
\begin{eqnarray*}
X^1_1&=& X^1_1,\\
X^1_2&=& \partial_1 X^1_1 + 2 \eta^{22}  \theta^{\frac{5}{2}} F_4 + \frac{4 (\eta^{22} )^2 \theta^4 F_2-2 \eta^{22}  \theta^3 \partial_1 F_2}{\eta^{12}}\\
&&    + 2 \eta^{12}  \theta X^2_1-2\eta^{22}  \theta X^1_1 + \theta F_1, \\
X^1_3&=& \partial_2 X^1_1 - \theta^3 \partial_1 F_2 + 3 \eta^{22}  \theta^4 F_2 + 2 \eta^{12}  \theta^{\frac{5}{2}} F_4,\\
X^1_4&=& -4 \eta^{12}  \theta^4 F_2,\\
X^1_5&=& \theta^3 F_2,\\
X^2_1&=& X^2_1,\\
X^2_2&=& \partial X^2_1 + F_3,\\
X^2_3&=& \partial_2 X^2_1 - \partial_1 X^1_1 + \frac{4 \eta^{22}  \theta^3 \partial_1 F_2- \theta^2 \partial_1^2 F_2- 6(\eta^{22} )^2  \theta^4 F_2}{\eta^{12}} \\
&& +\theta^{\frac{3}{2}} \partial_1 F_4-\frac{3}{2} \eta^{22}  \theta^{\frac{5}{2}} F_4,\\
X^2_4&=& 4 \theta^{3} \partial_1 F_2 -12 \eta^{22}  \theta^4 F_2-3 \eta^{12}  \theta^{\frac{5}{2}} F_4-\partial_2 X^1_1,\\
X^2_5&=& \frac{2 \eta^{22}  \theta^3 F_2 -\theta^2 \partial_1 F_2}{\eta^{12}} - X^1_1 + \theta^{\frac{3}{2}} F_4, 
\end{eqnarray*}
here $\theta=(2 \eta^{12}  u^2 + \eta^{22}  u^1)^{-1}$ and $F_i=F_i(u^1)$, for $i=1,\ldots,4$.
Choosing $R,S$ such that $X^i_1=Y^i_1$ for $i=1,2$, we can reduce $X^i_1$ to zero, obtaining
\begin{eqnarray*}
X^1_1&=& 0,\\
X^1_2&=& 2 \eta^{22} \theta \left(\theta^{\frac{3}{2}} F_4 -\frac{\partial_1 (\theta^2 F_2)}{\eta^{12}} \right)+ \theta F_1,\\
X^1_3&=& 2  \eta^{12} \theta^{\frac{5}{2}} F_4 - \partial_1 (\theta^3 F_2),\\
X^1_4&=& -4  \eta^{12} \theta^4 F_2,\\
X^1_5&=& \theta^3 F_2,\\
X^2_1&=& 0,\\
X^2_2&=& F_3,\\
X^2_3&=&\partial_1 ( \theta^{\frac{3}{2}} F_4)-\frac{\partial_1^2(\theta^2 F_2)}{\eta^{12}},\\
X^2_4&=& 4 \partial_1(\theta^3 F_2)+ \partial_2 ( \theta^{\frac{3}{2}} F_4),\\
X^2_5&=& \theta^{\frac{3}{2}} F_4 -\frac{ \partial_1(\theta^2 F_2)}{\eta^{12}}.
\end{eqnarray*}
Also in this case, the deformations depend on four functions on $u^1$.

\medskip
\noindent
\emph{Case 3:} $\kappa\neq 0, -1, -2$.
The condition $d_{\omega_1} d_{\omega_2}X=0$ for $\deg(X)=2$ implies
\begin{eqnarray*}
X^1_1&=& X^1_1,\\
X^1_2&=& \partial_1 X^1_1 + \frac{\kappa(\kappa+2)}{3(\kappa+1)^2}\theta^{\frac{\kappa-1}{2}} \partial_1 F_2 -\frac{\kappa(\kappa^2+7\kappa+4)\eta^{22}}{6(\kappa+1)} \theta^{\frac{\kappa-3}{2}} F_2 \\
&&+\theta^{-1} ( \kappa( \eta^{22} X^1_1 -\eta^{12} X^2_1) + F_1 ),\\
X^1_3&=&\partial_2 X^1_1-\frac{ \kappa(\kappa-1) \eta^{12}}{3(\kappa+1)}  \theta^{\frac{\kappa-3}{2}} F_2 ,\\
X^1_4&=& 0,\\
X^1_5&=& 0,\\
X^2_1&=& X^2_1,\\
X^2_2&=& \partial_2 X^2_1,\\
X^2_3&=& \partial_2 X^2_1 - \partial_1 X^1_1+ \theta^{\frac{\kappa-1}{2}} \partial_1 F_2 - \frac{1}{2} \eta^{22}(\kappa-1)(\kappa+1)\theta^{\frac{\kappa-3}{2}} F_2,\\
X^2_4&=&  (\kappa-1) \eta^{12} \theta^{\frac{\kappa-3}{2}} F_2-\partial_1 X^1_1,\\
X^2_5&=& \theta^{\frac{\kappa-1}{2}}F_2 - X^1_1,
\end{eqnarray*}
here $\theta= 2 \eta^{12} u^2 - (\kappa+1) \eta^{22} u^1$ and $F_i$ for $i=1,2$ are arbitrary functions depending on $u^1$.
Choosing $R,S$ such that $X^i_1=Y^i_1$ for $i=1,2$ we can remove $X^i_1$, obtaining
\begin{eqnarray*}
X^1_1&=& 0,\\
X^1_2&=& \frac{\kappa(\kappa+2)}{3(\kappa+1)^2}\theta^{\frac{\kappa-1}{2}} \partial_1 F_2  -\frac{\kappa(\kappa^2+7\kappa+4)\eta^{22}}{6(\kappa+1)} \theta^{\frac{\kappa-3}{2}} F_2 + \theta^{-1} F_1, \\
X^1_3&=& -\frac{ \kappa}{3(\kappa+1)}  \partial_2 ( \theta^{\frac{\kappa-1}{2}}F_2) ,\\
X^1_4&=& 0,\\
X^1_5&=& 0,\\
X^2_1&=& 0,\\
X^2_2&=& 0,\\
X^2_3&=& \partial_1 (\theta^{\frac{\kappa-1}{2}}F_2),\\
X^2_4&=& \partial_2 ( \theta^{\frac{\kappa-1}{2}}F_2),\\
X^2_5&=& \theta^{\frac{\kappa-1}{2}}F_2.
\end{eqnarray*}
In this last case, the deformations depend on two functions of $u^1$. The canonical Miura transformation reducing the pencil
 to the form described in the step 3 is generated by the Hamiltonian functional
\begin{gather*}
H=\int_{S^1}\left(-\frac{\theta^{\frac{\kappa -1}{2}}  (   4 \kappa \eta^{12} (\kappa-1) u^2 + \eta^{22} (\kappa+1) (2 \kappa^3 + 7 \kappa^2 +12 \kappa + 3) u^1 ) F_2   }{6 (\eta^{11})^2  (\kappa+1)^2(\kappa-1)}\right.\\
\left.+\frac{\theta^{\frac{\kappa +1}{2}}   (2 \kappa +3) u^1F_2'}{3 (\eta^{11})^2  (\kappa+1)^2}
+\frac{ \log \theta (\kappa+1) u^1 F_1}{2(\eta^{11})^2  \kappa }\right)
u^1_x\,dx.
\end{gather*}

\section{Appendix. Lift of Frobenius structures}\label{liftFrob}

Recall that a Frobenius manifold is a smooth manifold $M$ equipped with a pseudo-metric $g$ with Levi-Civita connection $\nabla$, a symmetric bilinear tensorial product on vector fields $\cdot$, and two vector fields $e, E$ such that
\begin{itemize}
\item $\nabla^\lambda_X Y = \nabla_XY + \lambda X \cdot Y$ defines a flat affine connection $\nabla^\lambda$ for all $\lambda \in \mathbf R$,
\item $\nabla e = 0$, $[e,E]=e$, and $e\cdot X = X$ for all vector fields $X$,
\item $\nabla(\nabla E) = 0$, $L_E\cdot=\cdot$, and $L_Eg=k g$ for some constant $k$.  
\end{itemize}
\begin{thm}\label{liftFrobenius}
Let $(M,g,\cdot,e,E)$ be a Frobenius manifold. Then the lifted tensors $\hat g, \hat \cdot, \hat e, \hat E$ define a structure of Frobenius manifold on $TM$. The Frobenius potential of the lifted structure is given by the lift of the Frobenius potential  $\hat F =v^i \frac{\partial F}{\partial u^i}.$

\end{thm}
\begin{proof}
From \eqref{liftg} one readily sees that $\hat g$ is symmetric and non-degenerate as soon as $g$ is.
If $\nabla$ is the Levi-Civita connection of $g$, then the lift $\hat \nabla$ is the Levi-Civita connection of $\hat g$.
This follows by uniqueness of Levi-Civita connection once one noticed that $\hat \nabla \hat g = 0$ and that $\hat \nabla$ is torsion free.
To see this notice that $\hat \nabla \hat g = 0$ for $\nabla g=0$, and that $\hat \nabla$ is torsion free by Proposition \ref{liftTR} and by torsion-freeness of $\nabla$. 

From \eqref{liftdot} is clear that $\hat \cdot$ is symmetric for $\cdot$ is.
Moreover, by definition of complete lift for connections it follows that $\hat \nabla^\lambda_XY = \hat \nabla_XY + \lambda X \hat \cdot Y$ for all $\lambda \in \mathbf R$, where now $X,Y$ are arbitrary tensor fields on $TM$. 
Thanks to Proposition \eqref{liftTR}, then $\hat \nabla^\lambda$ is flat.
All other conditions follows directly from definition of complete lift, and invariance of Lie derivative under complete lift.
\end{proof}

At this point  recall that a Frobenius manifold is said to be \textit{massive} if the algebra structure induced by the product $\cdot$ on any tangent space to $M$ is semisimple. More explicitly this means that there is no tangent vector $X$ on $M$ such that $X\cdot\ldots\cdot X=0$ for some finite product.
One may wonder whether semisemplicity assumption is preserved by complete lift or not.
In fact it is not, nor is possible to get a massive Frobenius manifold by complete lift of any Frobenius structure on $M$.
The reason is that any vector $Y$ which is tangent to the fibers of $TM$ is an idempotent for the algebra structure induced by $\hat \cdot$. 
Indeed any such vector has the local expression $Y^i \frac{\partial}{\partial y^i}$, whence it follows that $Y \hat \cdot Y=0$ thanks to \eqref{liftdot}.

\begin{rmk}
Given a Frobenius manifold  $(M,g,\cdot,e,E)$ one can define a hierarchy of quasilinear systems of PDEs of the form
$$u^i_{t_{p,\alpha}}=P^{ij}\frac{\delta H_{p,\alpha}}{\delta u^j},\qquad i=1,...,n,\,\,p=1,...,n,\,\,\alpha=0,1,2,3,...$$
where $P^{ij}$ is Hamiltonian operator of hydrodynamic type associated with the invariant metric $g$ and $ H_{p,\alpha}$ 
 are suitable local functionals in involution 
$$\{ H_{p,\alpha}, H_{q,\beta}\}_P=\int_{S^1}\frac{\delta H_{p,\alpha}}{\delta u^i}\left(g^{ij}\d_x+b^{ij}_ku^k_x\right)\frac{\delta H_{q,\beta}}{\delta u^j}\,dx=0$$
with respect to the associated Poisson bracket $\{,\}_P$. 
It is easy to check that the flows of the lifted hierarchy
$$u^i_{t_{p,\alpha}}=\hat{P}^{ij}\frac{\delta \hat{H}_{p,\alpha}}{\delta u^j},\qquad i=1,...,2n,\,\,p=1,...,n,\,\,\alpha=0,1,2,3,...$$
coincide with "half " of the flows of the principal hierarchy of the lifted Frobenius structure. The involutivity of the lifted
 Hamiltonian functionals 
$$\hat{H}_{p,\alpha}=\int_{S^1}v^{s}\d_s h_{p,\alpha}\,dx$$
follows from the identity  \eqref{main}. 
Indeed, due to this identity  any family of 1-forms in involution with respect to $\{\cdot,\cdot\}_P$ defines a family
 of Hamiltonians in involution with respect to $\{\cdot,\cdot\}_{\hat{P}}$. If the 1-forms are exact the Hamiltonians on the tangent
 bundle are the lift of the Hamiltonians on the base manifold. 

\end{rmk}

\section{Appendix. Lift of Hamiltonian vector fields}
Given a  Hamiltonian vector field $P\delta H$ with $\int_{S^1}h(u,u_x,...)\,dx$, 
we want to compare its complete lift 
$$
\widehat{P\delta H}=P\f{\delta H}{\delta u}\f{\d}{\d u}+\sum_kv_{(k)}\f{\d (P\f{\delta H}{\delta u})}{\d u_{(k)}}\f{\d}{\d v}
$$
with the vector field
$$
\hat{P}\delta\hat{H}=P\f{\delta H}{\delta u}\f{\d}{\d u}+\left(
P\f{\delta \hat{H}}{\delta u}+  \sum_{t}v_{(t)}\frac{\d P}{\d u_{(t)}} \f{\delta \hat{H}}{\delta v}\right)\f{\d}{\d v}
$$
where $\hat{H}[u,v]=\int_{S^1}v\f{\delta H}{\delta u}\,dx$. Since the components along $\f{\d}{\d u}$  coincide we have to show that
$$P\f{\delta \hat{H}}{\delta u}+  \sum_{t}v_{(t)}\frac{\d P}{\d u_{(t)}} \f{\delta \hat{H}}{\delta v}=\sum_kv_{(k)}\f{\d (P\f{\delta H}{\delta u})}{\d u_{(k)}}.$$
We observe that  
$$\f{\delta \hat{H}}{\delta v}=\f{\delta H}{\delta u},\qquad \f{\delta \hat{H}}{\delta u}=\f{\delta}{\delta u}\left(\sum_k\int_{S^1}v_{(k)}\f{\d h}{\d u_{(k)}}\,dx\right),$$
where the second identity has been obtained integrating by parts.
Using these facts and taking into account that the operators $\d_x$ and $\sum_{k}v_{(k)}\frac{\d}{\d u_{(k)}}$ commute, we get
\begin{eqnarray*}
&&P\f{\delta \hat{H}}{\delta u}+  \sum_{k}v_{(k)}\frac{\d P}{\d u_{(k)}} \f{\delta \hat{H}}{\delta v}=\\
&&P\f{\delta}{\delta u}\left(\sum_k\int_{S^1}v_{(k)}\f{\d h}{\d u_{(k)}}\,dx\right)+\sum_{k}v_{(k)}\frac{\d P}{\d u_{(k)}} \f{\delta H}{\delta u}=\\
&&P\sum_{h,k}(-1)^h\d_x^h\left(v_{(k)}\f{\d^2 h}{\d u_{(k)}\d u_{(h)}}\right)+\sum_{k}v_{(k)}\frac{\d P}{\d u_{(k)}} \f{\delta H}{\delta u}=\\
&&P\sum_kv_{(k)}\frac{\d}{\d u_{(k)}}\left[\sum_h(-1)^h\d_x^h\left(\f{\d h}{\d u_{(k)}}\right)\right]+\sum_{k}v_{(k)}\frac{\d P}{\d u_{(k)}} \f{\delta H}{\delta u}=\\
&&\sum_kv_{(k)}\f{\d (P\f{\delta H}{\delta u})}{\d u_{(k)}}.
\end{eqnarray*}
In the non scalar case the proof works in exactly the same way.

\section*{Ackowledgements} 
We would like to  thank Jenya Ferapontov and Raffaele Vitolo for useful discussions and  Joseph Krasil'shchik and Alik Verbovetsky
 for pointing out the reference \cite{Ku}. 
P.L.  is partially supported  by the Italian MIUR Research Project \emph{Teorie geometriche e analitiche dei sistemi Hamiltoniani in dimensioni finite e infinite} and by GNFM  Progetto Giovani 2014
 \emph{Aspetti geometrici e analitici dei sistemi integrabili}.

\end{document}